\documentclass[a4paper,USenglish,cleveref,autoref,thm-restate]{lipics-v2021}

\usepackage[utf8]{inputenc}
\usepackage[T1]{fontenc}
\usepackage{amsmath,amsfonts,amsthm,amssymb}
\usepackage{xfrac}					
\usepackage[sanserif,basic]{complexity}
\usepackage[algoruled, vlined, linesnumbered]{algorithm2e}
\usepackage{setspace}
\usepackage{paralist} 				

\usepackage[dvipsnames]{xcolor}     
\usepackage{tikz}                   
\usepackage{xspace}

\usepackage{footnote}
\makesavenoteenv{tabular}
\makesavenoteenv{table}

\crefname{case}{case}{cases}
\creflabelformat{stat}{#2{\upshape#1}#3}

\crefname{ineq}{inequality}{inequalities}
\creflabelformat{ineq}{#2{(\upshape#1)}#3} 

\crefname{alg}{algorithm}{algorithms}
\creflabelformat{alg}{#2{#1}#3}

\let\oldsqrt\sqrt
\def\hksqrt{\mathpalette\DHLhksqrt}
\def\DHLhksqrt#1#2{\setbox0=\hbox{$#1\oldsqrt{#2\,}$}\dimen0=\ht0
   \advance\dimen0-0.2\ht0
   \setbox2=\hbox{\vrule height\ht0 depth -\dimen0}%
   {\box0\lower0.4pt\box2}}
\renewcommand\sqrt\hksqrt

\renewcommand{\leq}{\leqslant}
\renewcommand{\geq}{\geqslant}
\renewcommand{\le}{\leqslant}
\renewcommand{\ge}{\geqslant}

\newcommand*{\Otilde}{\widetilde{O}}

\newcommand*{\D}{\mathcal{D}}
\newcommand*{\G}{\mathcal{G}}

\newcommand*{\T}{\mathcal{T}}
\newcommand*{\X}{\mathcal{X}}
\newcommand*{\Y}{\mathcal{Y}}
\newcommand{\inedges}{\textsc{In-Edges}}

\newcommand*{\nwspace}{\hspace*{.1em}} 

\DeclareMathOperator{\diam}{diam}
\DeclareMathOperator{\ecc}{ecc}

\author{Davide Bilò}%
{Department of Information Engineering, Computer Science and Mathematics,\\
	University of L'Aquila, Italy}%
{davide.bilo@univaq.it}%
{https://orcid.org/0000-0003-3169-4300} 
{} 

\author{Keerti Choudhary}%
{Department of Computer Science and Engineering, Indian Institute of Technology Delhi, India}%
{keerti@iitd.ac.in}%
{https://orcid.org/0000-0002-8289-5930} 
{} 

\author{Sarel Cohen}%
{School of Computer Science, The Academic College of Tel Aviv-Yaffo, Israel}%
{sarelco@mta.ac.il}%
{https://orcid.org/0000-0003-4578-1245} 
{} 

\author{Tobias Friedrich}%
{Hasso Plattner Institute, University of Potsdam, Germany}%
{tobias.friedrich@hpi.de}%
{https://orcid.org/0000-0003-0076-6308} 
{} 

\author{Martin Schirneck}%
{Hasso Plattner Institute, University of Potsdam, Germany}%
{martin.schirneck@hpi.de}%
{https://orcid.org/0000-0001-7086-5577} 
{} 

\authorrunning{Davide Bilò, Keerti Choudhary, Sarel Cohen, Tobias Friedrich, and Martin Schirneck}
\Copyright{Davide Bilò, Keerti Choudhary, Sarel Cohen, Tobias Friedrich, and Martin Schirneck}

\title{Deterministic Sensitivity Oracles for Diameter, Eccentricities and All Pairs Distances}



\nolinenumbers 

\hideLIPIcs  

\EventEditors{ }
\EventNoEds{0}
\EventLongTitle{ }
\EventShortTitle{ }
\EventAcronym{ }
\EventYear{ }
\EventDate{ }
\EventLocation{ }
\EventLogo{ }
\SeriesVolume{}
\ArticleNo{ }

\begin{document}

\maketitle

\begin{abstract}
	We construct data structures for extremal and pairwise distances in directed graphs
	in the presence of transient edge failures.
	Henzinger et al.~[ITCS 2017] initiated the study of fault-tolerant (sensitivity) oracles
	for the diameter and vertex eccentricities.
	We extend this with a special focus on space efficiency.
	We present several new data structures,
	among them the first fault-tolerant eccentricity oracle for dual failures
	in subcubic space.
	We further prove lower bounds that show limits
	to approximation vs.\ space and diameter vs.\ space trade-offs for fault-tolerant oracles.
	They highlight key differences between data structures for undirected and directed graphs.
	
	Initially, our oracles are randomized leaning on a sampling technique
	frequently used in sensitivity analysis.
	Building on the work of Alon, Chechik, and Cohen [ICALP 2019]
	as well as Karthik and Parter [SODA 2021],
	we develop a hierarchical framework to derandomize fault-tolerant data structures.
	We first apply it to our own diameter and eccentricity oracles
	and then show its versatility by derandomizing algorithms from the literature:
	the distance sensitivity oracle of Ren [JCSS 2022] and
	the Single-Source Replacement Path algorithm of Chechik and Magen [ICALP 2020].
	This way, we obtain the first deterministic distance sensitivity oracle with subcubic
	preprocessing time.
\end{abstract}

\section{Introduction}
\label{sec:intro}

The problems of computing {\em shortest paths}, {\em graph diameter}, and {\em vertex eccentricities} 
are fundamental in many applications of both theoretical
and applied computer science.
We address these problems in the setting of {\em fault tolerance}.
The interest in this problem setting stems from 
the fact that most real-world networks are prone to failures.
These are unpredictable but usually small in numbers and transient
due to some simultaneous repair process.
However, in an error-prone network, it is not always practical to recompute distances from scratch 
even if the number of edge failures is bounded.
A commonly adopted solution is that of designing {\em $f$-edge fault-tolerant oracles}, that is, compact data structures that can quickly report  exact or approximate extremal and pairwise distances in the network after up to $f$ edges failed. 
These structures are also known as \emph{sensitivity oracles}, where the sensitivity is
the maximum number $f$ of supported failures.

Many known fault-tolerant data structures are randomized.
The algorithm that preprocesses the underlying network may depend on random bits
or the correctness of the oracle's answers is only guaranteed with some probability.
Besides the practical difficulties of working with (true) randomness in computing, 
it is an interesting question to what extend randomness as a resource 
is needed to obtain efficient fault-tolerant oracles.
In this paper, we show that for a wide range of applications randomness can be removed
with only a slight loss of performance, or even none at all in some cases.
For this, we develop a novel derandomization framework and combine it with known 
techniques to obtain the following results.

\begin{itemize}
    \item We present new deterministic $f$-edge fault-tolerant oracles that report the
    exact/approximate diameter and vertex eccentricities in directed graphs  
    and we show lower bounds charting the limits
	of approximation vs.\ space and diameter vs.\ space trade-offs.
    
    \item 
    We derandomize the single-failure {\em distance sensitivity oracle} (DSO) of Ren~\cite{Ren22Improved} that can report exact distance for any pair of vertices in constant time. Our result gives the first deterministic exact DSO with truly sub-cubic processing time and constant query time.
    
    \item We derandomize the algorithm of Chechik and Magen~\cite{ChMa19} for the  {\em Single-Source Replacement Paths}  (SSRP) problem on directed graphs, that is, the task of finding a shortest path from a distinguished source vertex to every target, for every possible edge failure.
\end{itemize}

We believe that our techniques are of independent interest
and can help derandomize also other algorithms and data structures in the fault-tolerant domain.
Throughout the paper, the underlying network is modeled by a directed graph $G=(V,E)$, 
possibly with weights on its edges, where $V$ is the set of $n$ vertices and $E$ the set of $m$ edges.

\subsection{Diameter and Eccentricity Oracles in Directed Graphs}

In~\autoref{sec:FDOs_digraphs}, we discuss fault-tolerant oracles 
for the diameter and vertex eccentricities of a directed graph. 
We abbreviate {\em $f$-edge fault-tolerant diameter oracle} as $f$-FDO and
{\em $f$-edge fault-tolerant eccentricity oracle} as $f$-FEO.
In case of a single failure, $f = 1$, we shorten this to FDO and FEO, respectively.
The problem of designing FDOs was originally raised by Henzinger et al.~\cite{HenzingerL0W17}
and recently received some renewed interest by Bilò et al.~\cite{BCFS21DiameterOracle_MFCS}. 
Although the major focus of the latter work was on undirected graphs, the authors also showed that, for directed graphs, one can compute, in $\Otilde(mn + n^2/\varepsilon)$ time,\footnote{%
	For a non-negative function $g = g(n)$,
	we use $\Otilde(g)$ to denote $O(g \cdot \textsf{polylog}(n))$.
}  an oracle of size\footnote{%
	Unless stated otherwise, we measure the space in the number of $O(\log n)$-bit machine words.
}
$O(m)$ and constant query time that guarantees a {\em stretch} of $1+\varepsilon$, 
that is, it reports an upper bound on the value of the diameter within a factor of $1{+}\varepsilon$,
for any $\varepsilon > 0$.

Bilò et al.~\cite{BCFS21DiameterOracle_MFCS} also gave a complementary space lower bound
showing that any fault-tolerant diameter oracle with a sufficiently small stretch
must take $\Omega(m)$ bits of space.
However, this is not the full picture in that their construction only holds for diameter $2$.
We show here that in reality there is a transition happening:
the larger the diameter, the more space we can save,
up to a point where even $o(m)$-space oracles become possible.
We aim at pinpointing this transition,
starting with a generalization of the bound in \cite{BCFS21DiameterOracle_MFCS}
to diameter up to $n/\sqrt{m}$.

\begin{restatable}{theorem}{lowerbounddigraph}
	Let $n,m,D\geq 3$ be integers with $D = O(\frac{n}{\sqrt{m}})$.
	Any FDO with stretch $\sigma < \frac{3}{2}-\frac{1}{D}$
	on $n$-vertex, $m$-edge unweighted directed graphs of diameter $D$ requires $\Omega(m)$ bits of space,
	regardless of the query time. 
\label{theorem:lower_bound_digraph}
\end{restatable}

\noindent
Given an oracle for the fault-tolerant eccentricities with query time $q$,
one can emulate a diameter oracle with query time $nq$ by taking the maximum over all vertices.
The information-theoretic lower bound of~\Cref{theorem:lower_bound_digraph}
is independent of the query time and therefore every FEO also must have size $\Omega(m)$.

Notably,~\Cref{theorem:lower_bound_digraph} implies that, for any $0<\delta\leq 1$ and 
all digraphs with $n^{1+\delta}$ edges and a relatively small diameter of $O(\sqrt{n^{1-\delta}})$,
an FDO of stretch essentially $3/2$ takes $\Omega(n^{1+\delta})$ bits of space.
As hinted above,
this approximation vs. space trade-off no longer holds when we consider directed graphs with large diameter of $\omega(n^{5/6})$, for which we can design FDOs of quasi-linear (in $n$) space and negligible stretch. 

\begin{restatable}{theorem}{largediam}
	Let $G$ be a directed graph with $n$ vertices, $m$ edges, and diameter $D=\omega(n^{5/6})$
	and let $\varepsilon = \frac{n^{5/6}}{D}=o(1)$. 
	There is an FDO for $G$ with stretch $1+\varepsilon$, preprocessing time $\Otilde(mn)$,
	space $O(n \nwspace \log^2 n)$, and constant query time.
\label{theorem:large_diam_n^5/6}
\end{restatable}

The gap between the stretch-size trade-offs provided in~\Cref{theorem:lower_bound_digraph} and~\Cref{theorem:large_diam_n^5/6}, respectively, suggests that there must be a threshold between $n/\sqrt{m}$ and $n^{5/6}$ where low-stretch FDOs of sub-linear size and constant query time become possible.
We further narrow this gap and aim to find the smallest value for the diameter for which one can design an FDO with $o(m)$ space and constant query time.
We show that this is possible for directed graphs with 
diameter $\omega((n^{4/3}\log n)/\sqrt{m})$.
We leave it as an open problem to determine the smallest function $g$
such that directed graphs with diameter $g(n)/\sqrt{m}$ admit an FDO with $o(m)$ space.
Our results show that $g$ is of order $\omega(n)$ and $O(n^{4/3} \log n)$.

\begin{restatable}{theorem}{largediamupperbound}
	Let $G$ be be a directed graph with $n$ vertices, $m$ edges, and diameter $\omega((n^{4/3}\log n)/(\varepsilon\sqrt{m}))$.
	For any $\varepsilon = \varepsilon(n,m) > 0$,
	there is an FDO for $G$ with stretch $1+\varepsilon$, preprocessing time $\Otilde(mn)$,
	space $o(m)$, and constant query time.
\label{theorem:large_diam_upper_bound}
\end{restatable}

For the sake of readability, the FDOs in \autoref{sec:FDOs_digraphs} are randomized. Later, in~\autoref{sec:derandomization_technique}, we describe our derandomization framework and show how to apply it to both FDOs.

We now move our attention to the case multiple edge failures and give bounds
in terms of $f$ and $n$ on the minimum space requirement of $f$-FDOs.
Bilò et al.~\cite{BCFS21DiameterOracle_MFCS} designed an $f$-FDO for \emph{undirected} graphs of stretch $f+2$ that takes space $\Otilde(fn)$. The size of this oracle is optimal up to polylogarithmic factors.
In the next theorem, we show that such compact oracles are impossible for directed graphs,
even when allowing arbitrarily large stretch

\begin{restatable}{theorem}{lowerboundarbitrarystretch}
	Let $n,f$ be positive integers such that $2^{f/2}=O(n)$.
	Any $f$-FDO with an arbitrary finite stretch on $n$-vertex directed graphs
	requires $\Omega(2^{f/2} \nwspace n)$ bits of space, regardless of the query time.
\label{theorem:lower_bound_arbitrary_stretch}
\end{restatable}

The lower bound of~\autoref{theorem:lower_bound_arbitrary_stretch} marks an exponential-in-$f$ separation between the undirected and directed setting.
The directed graph used in the proof is inspired by the lower-bound construction used by Baswana et al.~\cite{BaswanaCR18} for the {\em $f$-edge fault-tolerant Single-Source Reachability} problem. This problem asks to compute the sparsest subgraph $H$ of a directed graph $G$ that preserves reachability from a designated source vertex $s$, that is, for every vertex $v$ and every set $F$ of $|F|\leq f$ edge failures, there is path from $s$ to $v$ in $H$ that avoids every edge in $F$
if and only if there is such a path in $G$. Baswana et al.~\cite{BaswanaCR18} provided a class of directed graphs for which any subgraph preserving single-source reachability with sensitivity $f$ has $\Omega(2^f n)$ edges. 
Our lower bound requires non-trivial extensions of their construction as it needs to satisfy several additional properties.
For example, the directed graph in~\cite{BaswanaCR18} has unbounded diameter,
while any lower bound for FDOs requires strongly connected graphs. 

We also consider the design of fault-tolerant eccentricity oracles
for general directed graphs as well as directed acyclic graphs (DAGs).
For the single-failure case and exact eccentricities, there is a folklore solution using the DSO of Bernstein and Karger~\cite{BeKa09} that runs in $\Otilde(n^3)$ time.
Henzinger et al.~\cite{HenzingerL0W17} showed how to trade stretch for running time
and presented an $(1{+}\varepsilon)$-approximate solution with preprocessing time
$\Otilde(mn + n^{3/2} \sqrt{Dm/\varepsilon})$, where $D$ denotes the diameter of the underlying graph.
Both oracles build a look-up table of size $O(n^2)$ using the fact that, for any vertex $v$,
only the failure of an edge on a shortest path tree rooted in $v$ can change the eccentricity of $v$.
The table allows for a constant query time but generalizing this to multiple failures $f \ge 2$
would take $\Omega(n^{f+1})$ space.
We show how to do better than that.
We give a meta-theorem that turns any exact or approximate DSO for pairwise distances
into an FEO for eccentricities.
Plugging in any compact DSO for multiple failures then immediately gives
a space improvement for the FEO.
In the following, with stretch $\sigma = 1$, we mean exact oracles.

\begin{restatable}{theorem}{meta}
\label{theorem:meta}
	Let $G$ be a (undirected or directed and possibly edge-weighted) graph 
	with $n$ vertices and $m$ edges.
	Given access to a DSO for $G$ with sensitivity $f$, stretch $\sigma \ge 1$,  
	preprocessing time $P$, space $S$, and query time $Q$,
	one can construct an $f$-FEO for $G$ with stretch $1+\sigma$, preprocessing time $O(mn + P)$,
	space $O(n+S)$, and $O(f \cdot Q)$ query time.
\end{restatable}

There are multiple distance oracles to choose from,
all with different strengths and weaknesses.
When using the DSO for of Duan and Pettie~\cite{DP09},
we get in polynomial time a $2$-approximate $2$-FEO with space $O(n^2 \log^3 n)$.
To the best of our knowledge, this is the first eccentricity oracle for dual failures in subcubic space.
Van den Brand and Saranurak~\cite{BrSa19} gave a DSO supporting an arbitrary number of failures $f$.
On directed graphs with integer edge weights in the range $[-M,M]$
it has polynomial space and preprocessing time,
but a query time that depends both on the sensitivity and the graph size.
Let $\omega<2.37286$ be the matrix multiplication exponent~\cite{AlmanVWilliams21RefinedLaserMethod}.
Plugging the DSO in~\cite{BrSa19} into our reduction gives 
an $f$-FEO with stretch $2$, $O(M n^3)$ space\footnote{%
	In~\cite{BrSa19}, the space of the DSO is phrased as $O(M n^3 \log n)$ \emph{bits}.
},
and query time $O(M n f^{\omega+1})$.
On undirected graphs, we can make the query time independent of $n$
by applying the very recent DSO by Duan and Ren~\cite{DuanRen21ExactDistancesMultipleFailures_STOC}
with $O(fn^4)$ space and a query time of $f^{O(f)}$.
However, the preprocessing of the latter is only polynomial for constant $f$.
Since our reduction also applies to approximate oracles,
we get, for any $f = o(\log n/\log\log n)$ and $\varepsilon > 0$,
an $f$-FEO in polynomial time with stretch $(2+\varepsilon)$,
space $O(n^2 ((\log n)/\varepsilon)^f f)$ and query time $O(f^6 \log n)$
via the DSO by Chechik et al.~\cite{ChCoFiKa17}.

As already mentioned above, the $\Omega(2^{f/2} \nwspace n)$-bits lower bound
in \autoref{theorem:lower_bound_arbitrary_stretch} also holds for FEOs.
On DAGs, however, we can improve upon this and obtain a space requirement that is reminiscent of the one Bilò et al.~\cite{BCFS21DiameterOracle_MFCS} gave for \emph{undirected} graphs.
Note that in a DAG at most one vertex can have bounded eccentricity.

\begin{restatable}{theorem}{eccDAGupperbound}
\label{theorem:ecc_DAG_upper_bound}
	Let $G$ be a directed acyclic graph with, $m$ real-weighted edges, $n$ vertices, and a
	distinguished source vertex $s$. 
	For any integer $f$, there is an $f$-FEO for $G$ with stretch $f$, preprocessing time $\Otilde(m)$,
	space $O(nf)$, and $O(f)$ query time.
\end{restatable}

All the results for $f$-FDOs and $f$-FEOs are presented for edge failures.
However, they also hold for vertex failures using well-known transformation techniques for directed graphs.\footnote{%
	Indeed, we can transform the directed graph $G$ into some graph $G'$ with $2n$ vertices. 
	We represent each vertex $v$ of $G$ with an edge $(v^-,v^+)$ in $G'$, 
	and replace each edge $(u,v)$ of $G$ with the edge $(u^+,v^-)$ in $G'$.
	For edge-weighted $G$, the weight of the new vertex-edge is set to $0$ keeping eccentricities.
	For unweighted $G$, the eccentricity of $v$ in any subgraph $H$ of $G$ is half the eccentricity of $v^-$ in $H' \subseteq G'$.}

\subsection{Derandomization Technique}
\label{subsec:intro_derandomize}

We now turn to the derandomization of fault-tolerant data structures.
In \autoref{sec:derandomization_technique}, we develop the {\em Hierarchical Double Pivots Hitting Sets} (HDPH) algorithm as the center piece of a framework to derandomize known replacement paths algorithms and oracles. The aim of the HDPH algorithm is to compute a sequence of sets $B_1, \dots, B_{\log n} \subseteq V(G)$ such that each $B_i$ has size $\Otilde(n/2^i)$ and hits a set ${\cal P}_i$ of (replacement) paths each of length $\Omega(2^i)$. Unfortunately, the paths ${\cal P}_i$ that need to be hit by $B_i$ are not known in advance. Our algorithm fixes this issue by iteratively
computing the set of paths ${\cal P}_i$ using the previous sets $B_0,\dots,B_{i-1}$. 
The algorithm relies on the ability of the oracle we want to derandomize to be {\em path-reporting}, that is, to report a path representing the exact or approximate distance between the queried vertices for DSOs, the diameter for FDOs, or the vertex eccentricity for FEOs. 
We show how to implement the HDPH algorithm to derandomize 
the FDOs in \Cref{theorem:large_diam_n^5/6,theorem:large_diam_upper_bound},
the DSO of Ren~\cite{Ren22Improved} for directed graphs with integer edge weights in the range $[1,M]$, and the algorithm of Chechik and Magen~\cite{ChMa19} for the SSRP problem in directed graphs. 

\subparagraph{Distance Sensitivity Oracles.}
The concept of DSOs was introduced by Demetrescu et al.~\cite{DeThChRa08} who showed how to compute an exact DSO of size $O(n^2 \log n)$ and constant query time in $\Otilde(mn^2)$ time. Later, Bernstein and Karger~\cite{BeKa09} improved the preprocessing time to $\Otilde(mn)$ and Duan and Zhang~\cite{DuanZ17a} reduced the space to $O(n^2)$, which is asymptotically optimal. Algebraic algorithms are known to further improve the preprocessing times, if one is willing to employ fast matrix multiplication, see~\cite{ChCo20,GuRen21ConstructingDSO_ICALP} and the references therein.
For more results on approximate DSOs for both single and multiple failures, see~\cite{ChCoFiKa17,DuanGR21,DP09}.

We combine the HDPH framework with a recent breakthrough result by Karthik and Parter~\cite{KarthikParter21DeterministicRPC}
to derandomize the path-reporting DSO of Ren~\cite{Ren22Improved} for directed graphs with integer edge weights in the range $[1,M]$.
This was the first DSO that achieved a constant query time with a randomized subcubic preprocessing time
of $O(Mn^{2.7233})$.
On undirected graphs, the preprocessing improves to $\Otilde(Mn^{(\omega + 3)/2}) = O(Mn^{2.6865})$.
Our derandomization of Ren's DSOs in both settings incurs a slight loss of efficiency.
Nevertheless, we obtain the first deterministic DSO with constant query time and truly sub-cubic preprocessing. This improves significantly over the result by Alon, Chechik, and Cohen~\cite{AlonChechikCohen19CombinatorialRP} who designed a DSO with $O(mn^{4-\alpha})$ preprocessing time
and  $\Otilde(n^{2\alpha})$ query time, for any $\alpha \in (0,1)$.

\begin{theorem}
\label{thm:deterministic_DSO}
	For any $n$-vertex directed graph $G$ with integer edge weights in the range $[1,M]$,
	there exists a deterministic path-reporting DSO with $O(Mn^{2.8068})$ preprocessing time 
	and constant query time.
	If $G$ is undirected, the preprocessing time decreases to $\Otilde(Mn^{(\omega + 6)/3}) = O(Mn^{2.7910})$.
\end{theorem}

Recently, Gu and Ren~\cite{GuRen21ConstructingDSO_ICALP} presented a new randomized DSO 
with a preprocessing time of $O(Mn^{2.5794})$. Unfortunately, our HDHP algorithm cannot be used to derandomize it for the following two reasons. First, the DSO of Gu and Ren is not path-reporting. Secondly, it internally relies on probabilistic polynomial identity testing. It is a long-standing open question how to derandomize this, far beyond the field of fault-tolerant data structures.

\subparagraph{Single Source Replacement Paths Problem.}
In the SSRP problem we want to compute replacement paths from a designated source to each destination vertex, under each possible edge failure. Grandoni and Vassilevska Williams~\cite{GW12,GrandoniVWilliamsFasterRPandDSO_journal} first developed an algorithm for both directed and undirected graphs with integer edge weights in the range $[1,M]$ that uses fast matrix multiplication and runs in $\Otilde(Mn^\omega)$ time. Chechik and Cohen~\cite{ChechikCohen19SSRP_SODA} presented an $\Otilde(m\sqrt{n} + n^2)$ time
SSRP algorithm for \emph{undirected} graphs that was later simplified and generalized to deal with multiple sources by Gupta~et~al.~\cite{GJM20}. 
In this paper we use our HDPH framework to  derandomize the recent $\Otilde(m\sqrt{n}\,{+}\,n^2)$ time randomized algorithm for \emph{directed} graphs developed by Chechik and Magen~\cite{ChMa19}, without any loss in the time complexity. Specifically, we prove the following result.


\begin{restatable}{theorem}{singlesourcereplacementpaths}
\label{thm:single_source_replacement_paths}
	There exists a deterministic algorithm for the Single Source Replacement Path problem 
	in unweighted directed graphs
	running in time $\Otilde(m\sqrt n + n^2)$.
\end{restatable}

\section{Preliminaries}
\label{sec:prelims}

We let $G = (V,E)$ denote a directed graph on $n$ vertices and $m$ edges,
potentially edge-weighted by some function $w \colon E \to \mathbb{R}$. We tacitly assume that $G$ is strongly connected, in particular, $m = \Omega(n)$.
For any (weighted) directed graph $H$ (possibly different from $G$),
we denote by $V(H)$ and $E(H)$ the set of its vertices and edges, respectively.
Let $P$ be a path in $H$ from $s \in V(H)$ to $t \in V(H)$,
we say that $P$ is an \emph{$s$-$t$-path} in $H$. 
We denote by $|P| = \sum_{e \in E(P)} w(e)$ the \emph{length} of $P$,
that is, its total weight.
If $H$ is unweighted, we let $|P| = |E(P)|$ denote the number of its edges.
For $u,v \in V(P)$, we let $P[u..v]$ denote the subpath of $P$ from $u$ to $v$.
For $s,t \in V(H)$, the \emph{distance} $d_H(s,t)$ 
is the minimum length of any $s$-$t$-path in $H$;
if $s$ and $t$ are disconnected, we set $d_H(s,t) =+ \infty$.
When talking about the base graph $G$, we drop the subscripts
if this does not create any ambiguities.
The \emph{eccentricity} of a vertex $s \in V(H)$ is $\ecc_H(s) = \max_{t \in V(H)} d_H(s,t)$,
the \emph{diameter} is $\diam(H) = \max_{s \in V(H)} \ecc_H(s)$.
For a set $F \subseteq E(H)$ of edges,
let $H-F$ be the graph obtained from $H$ by removing all edges in $F$.
A \emph{replacement path} $P_H(s,t,F)$ is a shortest path from $s$ to $t$ in $H-F$. 
Its length $d_H(s,t,F) = |P_H(s,t,F)|$ is the \emph{replacement distance}.
The \emph{fault-tolerant eccentricity} of a vertex $s \in V$ of the base graph with respect to $F$ is 
$\ecc_{G\,{-}\,F}(s)$, the \emph{fault-tolerant diameter} is $\diam(G-F)$.

For a positive integer $f$,
an \emph{$f$-edge fault-tolerant eccentricity oracle} ($f$-FEO) for $G$
reports, upon query $(s,F)$ with $|F| \le f$, the value $\ecc_{G{-}F}(s)$.
An \emph{$f$-edge fault-tolerant diameter oracle} returns $\diam(G-F)$ upon query $F$.
For a single edge failure, we write FEO for 1-FEO and abbreviate $F = \{e\}$ to $e$.
For any real number $\sigma = \sigma(n,m,f) \ge 1$, an $f$-FEO is said to have \emph{stretch} $\sigma$,
or be \emph{$\sigma$-approximate},
if the returned value $\widehat \ecc(s,F)$ on query $(s,F)$ satisfies
$\ecc_{G-F}(s) \le \widehat \ecc(s,F) \le \sigma \cdot \ecc_{G-F}(s)$,
analogously for $f$-FDOs.
The \emph{preprocessing time} is the time needed to compute the data structure,
its \emph{query time} is the time needed to return an answer.
For weighted graphs, we assume the weight function being such 
that all distances can be stored in a single word on $O(\log n)$ bits.
Unless stated otherwise, we measure the \emph{space} of the oracles in the number of words.
The oracles cannot access features of graph $G$ 
except those stored during preprocessing.
The size of the input does not count against the space of the data structures.

\section{Diameter and Eccentricity Oracles}
\label{sec:FDOs_digraphs}

This section discusses fault-tolerant oracles for the diameter and vertex eccentricity in directed graphs.
We start by presenting space lower bounds for FDOs that guarantee a certain stretch
when supporting single or multiple edge failures, respectively.
In the single-failure case, the bound depends on the diameter of the graph.
Roughly speaking, if the base graph has low diameter, 
we cannot save much space over just storing all edges.
The picture changes if the diameter grows larger.
We show that then we can obtain FDOs with $o(m)$ space, or even $\Otilde(n)$.
We then turn the discussing to eccentricity oracles for dual and multiple failures.
Note that there we need to report not only one value per graph $G-F$,
but one per vertex in each of those, so special techniques are needed to handle the space increase.

\subsection{Space Lower Bounds for Diameter Oracles}
\label{subsec:FDOs_space_lower_bound_single_failure}

Bilò et al.~\cite{BCFS21DiameterOracle_MFCS} showed that any FDO with a stretch $\sigma < 3/2$
on \emph{undirected} $m$-edge graphs must take $\Omega(m)$ bits of space.
In particular, any data structure that can distinguish between a fault-tolerant diameter 
of $2$ and $3$ has this size.
Their construction transfers to directed graphs 
(by merely doubling each undirected edge into two directed ones).
However, they do not parameterize the graphs by their diameter,
namely, if the FDO has to distinguishing between diameter $D$ and $3D/2$ for some $D \ge 3$.
We generalize their result by showing that there is an intermediate range of $D$
where the $\Omega(m)$-bit bound still applies.
However, here the situation is more intricate in that large values of $D$
do allow for significant space reductions.

The construction\footnote{%
	The graph used in the proof of \cite[Lemma~12]{BCFS21DiameterOracle_MFCS}
	has the property that failing any edge can increase the diameter by at most $1$.
} 
by Bilò et al.~\cite{BCFS21DiameterOracle_MFCS}
cannot be extended to $D \ge 3$.
We introduce a new technique in the next lemma
whose proof is deferred to \Cref{subapp:proofs_omitted_Sec_3}.

\begin{restatable}{lemma}{lemmalowerbounddigraph}
\label{lemma:lower_bound_digraph}
	Let $n,m,D$ be integers such that $n^2 \ge m \ge n \ge 4$, 
	and $n/\sqrt{m} > D \ge 3$.
	There exists a family $\G$ of $n$-vertex directed graphs with diameter $D$ and $\Theta(m)$ edges 
	such that any data structure for graphs in $\G$ that decides
	whether the fault-tolerant diameter remains at $D$ or
	increases to $(3D{-}1)/2$ for odd $D$ (or $(3D/2) - 1$ for even $D$)
	requires $\Omega(m)$ bits of space.
\end{restatable}

It is now easy to obtain the $\Omega(m)$-bit lower bound of \autoref{theorem:lower_bound_digraph}
since any FDO of stretch $\sigma < \frac{3}{2} - \frac{1}{D}$ must tell the two cases apart.

We now turn to diameter oracles that support more than one edge failure, $f\,{>}\, 1$.
\autoref{theorem:lower_bound_arbitrary_stretch} states that they require 
space that is exponential in $f$,
even if we allow the stretch and query time to be arbitrarily large (but finite).
It follows from the next lemma together with the observation that such $f$-FDOs
have to detect whether the edge failures disconnect the graph.

\begin{restatable}{lemma}{connectivitylowerbound}
\label{lemma:connectivity_lower_bound}
	Any data structure for $n$-vertex digraphs 
	that decides for at most $2f = O(\log n)$ edge failures
	whether the fault-tolerant diameter is finite requires $\Omega(2^f n)$ bits of space.
\end{restatable}

\subsection{Improved Upper Bounds}
\label{subsec:FDO_upper_bound}

The above discussion shows that for graphs with \emph{small} diameter,
there is no hope to obtain an FDO whose space is much smaller 
than what is needed to store the full graph.
At least not while retaining good stretch at the same time.
The lower bound in \autoref{theorem:lower_bound_digraph},
however, breaks down for a \emph{large} diameter.
Indeed, we show next that in this regime we can do much better in terms of space,
without sacrificing stretch or query time.

\Cref{theorem:large_diam_n^5/6,theorem:large_diam_upper_bound} will follow 
from the same construction.
The initial way we present it in \autoref{lemma:large_diam} uses randomization
in the form of a well-known sampling lemma,
see \cite{GrandoniVWilliamsFasterRPandDSO_journal,RodittyZwick12kSimpleShortestPaths}.
We will later discuss how to derandomize the oracles.

\begin{lemma}[Sampling Lemma]
\label{lemma:sampling}
	Let $H$ be a $n$-vertex directed graph, $c > 0$ a positive constant, and $L \ge c \ln n$.
	Define a random set $B \subseteq V(H)$ by sampling each vertex of $H$ independently
	with probability $(c \ln n)/L$.
	With probability at least $1-\frac{1}{n^c}$, the cardinality of $B$ is $O((n \log n)/L)$.
	Let further ${\mathcal P}$ be a set of $\ell$ simple paths in $H$, 
	each of which spans $L$ vertices.
	With probability at least $1-\frac{\ell}{n^c}$, we have $V(P) \cap B \neq \emptyset$ for every $P \in {\mathcal P}$.
\end{lemma}

\begin{lemma}
	For any $n$-vertex, $m$-edge unweighted directed graph $G$ with diameter $D = \omega(\log n)$
	and any $\varepsilon = \varepsilon(n,m,D) > 0$,
	we can compute in time $\Otilde(mn+n^4/(\varepsilon^3 D^3))$ an FDO
	with $1+\varepsilon$ stretch, $O(n+(n^{8/3} \log^2 n)/(\varepsilon^2 D^2) )$ space,
	and constant query time.
\label{lemma:large_diam}
\end{lemma}

\begin{proof}
Let $D = \diam(G)$, $b=n/(\varepsilon D)$, and $c > 0$ a sufficiently large constant.
We sample a set $B\subseteq V$ of \emph{pivots}
by including each vertex independently with probability $(2bc \ln n)/n$.
By \autoref{lemma:sampling} with $L = n/2b = \varepsilon D/2$, there are $O(b \log n)$ many pivots w.h.p.

For the graph $G$, compute in $\Otilde(mn)$ time the $O(1)$-query time
distance sensitivity oracle of Bernstein and Karger~\cite{BeKa09}.
We further compute a subgraph $H$ of $G$ that is just the union of $|B|$ shortest-path trees,
one rooted at each pivot.
We iterate over the edges of $H$
and compute the collection $\X$ of all those $e \in E(H)$
such that $d(b_1,b_2,e) > d(b_1,b_2)$ for some pair $(b_1,b_2)\in B\times B$. 
The time to compute $\X$ is $O(n|B|^3) = \Otilde(n^4/(\varepsilon^3 D^3))$
since processing an edge in $H$ requires $|B|^2$ calls to the DSO. 
Observe that any subgraph of $G$ that exactly preserves distances between all pairs in $B\times B$ must contain all the edges of $\X$.
Bodwin~\cite{Bodwin17} showed that there are distance-preserving subgraphs with respect to $B\times B$ with at most $O(n+n^{2/3}|B|^2)$ edges. 
Thus, the size of $\X$ is bounded by $O(n+n^{2/3}|B|^2)$.

Next, we build a dictionary $\D_{\X}$ in which we store the the edges in $\X$ 
together with the maximum distance between any pair of pivots
if the edge fails (or $\diam(G)$ if this is larger). 
In other words, for each $e\in \X$,
we store $\phi(e)=\max\{ \max_{b_1,b_2\in B} \nwspace d(b_1,b_2,e), \nwspace \diam(G)\}$. 
Let $\Y$ be the set of all edges in $E$ such that $G{-}e$ is no longer strongly connected. We build a dictionary $\D_{\Y}$ in which we store information about the edges $\Y$. 
It is well-known that $\Y$ contains $O(n)$ edges 
and can be computed in time $O(m)$~\cite{Italiano12FindingStrongBridges}. 

Recall that $b = n/(\varepsilon \diam(G))$.
The oracle's output $\widehat{D}(e)$ is defined as follows:
if $e\in \Y$, then $\widehat{D}(e) = \infty$; 
if $e\in \X$, $\widehat{D}(e) =\phi(e)+n/b$; 
otherwise, the oracle outputs $\widehat{D}(e) = \diam(G) +n/b = (1+\varepsilon) \diam(G)$.

Evidently, the oracle is correct for all $e \in \Y$.
It is also easy to verify that all outputs are at most
$\phi(e) + n/b \le \diam(G-e)+n/b 
	= \diam(G-e) + \varepsilon \diam(G) \le (1+\varepsilon) \diam(G-e)$.
To prove that they are also at least $\diam(G-e)$,
consider a vertex pair $(u,v)\in V\times V$ 
such that $d(u,v,e)=\diam(G-e)<\infty$.
With high probability\footnote{
	We say an event occurs \emph{with high probability} (w.h.p.)
	if it has success probability $1-n^{-c}$ for some constant $c>0$
	that can be made arbitrarily large.
}
by \autoref{lemma:sampling},
there exists a shortest \mbox{$u$-$v$}-path in $G-e$ and two pivots $b_u,b_v\in B$ on that path
such that $d(u,b_u,e), d(b_v,v,e) \le L = n/2b$. 
We have $\diam(G-e)=d(u,b_u,e)+ d(b_u,b_v,e) + d(b_v,v,e)$.
Suppose $e \notin \X$.
Then, $d(b_u,b_v,e) = d(b_u,b_v) \le \diam(G)$ holds
and therefore $\diam(G-e) \le \diam(G) + n/b = \widehat{D}(e)$.
If $e \in \X$, then $d(b_u,b_v,e) \le \phi(e)$ and 
$\diam(G-e) \le \phi(e) + n/b = \widehat{D}(e)$.

There are $k$-element dictionaries of size $O(k)$ and $O(1)$ query time 
computable in time $\Otilde(k)$~\cite{HagerupMiltersenPagh01DeterministicDictionaries}.
The dictionaries have total size
$O(n+n^{2/3}|B|^2) = O(n + (n^{8/3}\log^2 n)/(\varepsilon^2 D^2))$.
\end{proof}

The oracle in \autoref{lemma:large_diam} can also be extended to handle vertex failures. 
The only modification required is to add to set $\X$ those vertices $v\in V$ that satisfy $d(b_1,b_2,v)  > d(b_1,b_2)$ for some $(b_1,b_2)\in B\times B$, 
and to add to $\Y$ to be those vertices $v$ for which $G-v$ is not strongly connected.
Suppose $D = \omega(n^{5/6})$, inserting any $\varepsilon \ge n^{5/6}/D = o(1)$ above gives an FDO
with near linear space and $1+o(1)$ stretch
that is computable in time $\Otilde(mn)$, which proves \autoref{theorem:large_diam_n^5/6}.
Furthermore, for graphs with diameter $\omega((n^{4/3}\log n)/(\varepsilon \sqrt{m}))$,
we obtain in $\Otilde(mn)$ time an FDO with constant query time and $o(m)$ space (\autoref{theorem:large_diam_upper_bound}).

\subsection{Eccentricity Oracles}

We now prove \autoref{theorem:meta}
that constructs an $f$-edge fault-tolerant eccentricity oracle
from a DSO supporting $f$ failures. The improved $f$-FEO for DAGs can be found in \Cref{subapp:ecc_oracle}.

Let $\D$ be a DSO with sensitivity $f$ and  stretch $\sigma$ that,
on (un-)directed possibly weighted graphs, 
can be computed in time $P$, uses $S$ space, and has a query time of $Q$. 
For any given source $s\in V$ and query set $F$ of $|F| \le f$ edges, 
our oracle reports an $(1{+}\sigma)$-approximation of the eccentricity of $s$ in $G-F$. 
We simply store $\D$ and, for each $x\in V$, the value $ecc_G(x)$.
All eccentricities in the base graph $G$ can be obtained with a BFS from each vertex in $O(mn)$.

Upon query $(s, F=\{(x_1,y_1),\ldots,(x_f,y_f)\})$,
we use $\D$ to compute $d(s,y_i,F)$, for all $1 \le i \le f$.
Our estimate is $\widehat{\ecc}_{G-F}(s) = \ecc_G(s)+\max_{1\leq i\leq f}d(s,y_i,F).$
%
The time taken to compute $\widehat{\ecc}_{G-F}(s)$ is $O(f\cdot Q)$
and the space requirement of the oracle is $O(n+ S)$.

Now we show that $\widehat{\ecc}_{G-F}(s)$ is a $(1+\sigma)$-approximation of ${\ecc}_{G-F}(s)$. Let $F_0$ be the subset of $F$ consisting of those edges in $F$ that lie on some shortest-path tree $T$ rooted in $s$. 
If $F_0$ is empty, we immediately get $\ecc_{G-F}(s) = \ecc_G(s) \le \widehat{\ecc}_{G-F}(s)$.
Otherwise, for any $v\in V$, either $d(s,v)=d(s,v,F)$ or there exists an $(x,y)\in F_0$ such  
that $y$ is an ancestor of $v$ in $T$. In this latter case $d(y,v,F) \leq {\ecc}_{G}(s)$. This proves that  
$d(s,v,F)\leq d(s,y,F)+d(y,v,F)\leq d(s,y,F)+{\ecc}_{G}(s)\leq \widehat{\ecc}_{G-F}(s)$.
Thus, ${\ecc}_{G-F}(s)\leq \widehat{\ecc}_{G-F}(s)$. Next observe that $\ecc_G(s) \leq {\ecc}_{G-F}(s)$ and $\max_{1\leq i\leq f}d(s,y_i,F) \leq \sigma\cdot {\ecc}_{G-F}(s)$, which proves that $\widehat{\ecc}_{G-F}(s)\leq (1+\sigma) \cdot {\ecc}_{G-F}(s)$.





\section{Derandomization Framework}
\label{sec:derandomization_technique}

The fault-tolerant diameter oracles in 
\Cref{theorem:large_diam_n^5/6,theorem:large_diam_upper_bound} are randomized.
They both follow from \autoref{lemma:large_diam}
which in turn relies on a random hitting set to intersect all replacement paths of a certain length.
In fact, many more data structures and algorithms in the fault-tolerant setting follow a sampling-based approach similar to \autoref{lemma:sampling}, 
see e.g.~\cite{ChechikCohen19SSRP_SODA,ChMa19,GrandoniVWilliamsFasterRPandDSO_journal,Ren22Improved,%
	RodittyZwick12kSimpleShortestPaths,WY13}.
It is an interesting question whether these algorithms can be derandomized efficiently.
Currently there is no single approach to derandomize \autoref{lemma:sampling}
in the same $O(n)$ time it uses to go through all vertices.
Therefore, the literature focuses on the specific applications.
The goal is to replace the sampling step by a deterministic construction of the hitting set
that, while taking $\omega(n)$ time, does not (or only marginally) 
increase the asymptotic running time of the whole algorithm.
Recently, there was some progress on notable special cases. 
Karthik and Parter~\cite{KarthikParter21DeterministicRPC} gave a derandomization
for the algebraic version of the distance sensitivity oracle of Weimann and Yuster~\cite{WY13}
with a slightly higher running time
(for a detailed discussion see \autoref{lem:KarthikParter} below).
Bilò et al.~\cite{BCFS21SingleSourceDSO_ESA} derandomized
the SSRP algorithms of Grandoni and Vassilevska Williams~\cite{GrandoniVWilliamsFasterRPandDSO_journal} as well as
Chechik and Cohen~\cite{ChechikCohen19SSRP_SODA}.
Their derandomization succeeds in the same time bounds as the original randomized algorithm,
but the technique only works for undirected graphs.
Here, we develop a framework for directed graphs.
We first apply it to our own FDOs
and then show its versatility by also derandomizing the DSO of Ren~\cite{Ren22Improved}
and the SSRP algorithm of Chechik and Magen~\cite{ChMa19}.

We build on the work of Alon, Chechik, and Cohen~\cite{AlonChechikCohen19CombinatorialRP}.
We first review some technical details of their result and then describe our additions.
For now, we assume the base graph $G$ to be unweighted
and only later (in \autoref{sec:derandomizing_existing})
incorporate positive integer edge weights.
For concreteness, consider the task in \autoref{lemma:large_diam}
of finding a set $B \subseteq V$, the pivots,
such that for all $s,t \in V$ and edge $e \in E$ with replacement distance $d(s,t,e)$ at least $L = \varepsilon \diam(G)/2$,
there exists some replacement path $P(s,t,e)$ that contains a pivot $x \in B$.
Other fault-tolerant algorithms pose similar requirements on $B$.
The technique in~\cite{AlonChechikCohen19CombinatorialRP}
consists of computing a small set of \emph{critical} paths, 
much smaller than the set of all $O(mn^2)$ replacement paths.
Once we have those, a hitting set can be computed with the folklore greedy algorithm,
called \textsf{GreedyPivotSelection} in~\cite{AlonChechikCohen19CombinatorialRP},
that always selects a vertex that is contained in the most unhit paths.\footnote{%
	To achieve the performance of \autoref{lemma:greedy-correctness},
	one has to truncate all paths by selecting $L$ vertices arbitrarily from each $P \in \mathcal{P}$.
	This is non-issue for us as, by construction, all our paths will have length $\Theta(L)$. 
}
Alternatively, one can use the blocker set algorithm of King \cite{King99FullyDynamicAPSP}.

\begin{lemma}[Alon, Chechik, and Cohen~\cite{AlonChechikCohen19CombinatorialRP}]
\label{lemma:greedy-correctness}
Let $1 \le L \le n$ and $1 \le q = \emph{\textsf{poly}}(n)$ be two integers.
Let $P_1, \ldots, P_q \subseteq V$ be sets of vertices
that, for every $1 \le k \le q$, satisfy $|P_k| \ge L$.
The algorithm \emph{\textsf{GreedyPivotSelection}} computes in time
$\Otilde(qL + n^2/L)$ a set $B \subseteq V$
of $|B| = O((n \log q)/L) = \Otilde(n/L)$ pivots
such that, for every index $k$, it holds that $B \cap P_k \ne \emptyset$.
\end{lemma}

The crucial part is to quickly find the paths $P_k$
such that hitting them is sufficient to hit all long replacement path.
Of course, this could be done by computing all-pairs shortest paths in each graph $G{-}e$
in total time $\Otilde(m^2n)$ using Dijkstra's algorithm
(or $\Otilde(mn^{2.5302})$ if one is willing to use fast rectangular matrix multiplication~\cite{LeGall12RectangularMatrixMultiplication,Zwick02DirectedAPSP}).
However, this is much more than the
$\Otilde(mn+n^4/(\varepsilon^3 \diam(G)^3))$ time bound we had in \autoref{lemma:large_diam}.
For the applications in~\cite{AlonChechikCohen19CombinatorialRP}, 
a single set of paths and therefore a single hitting set was sufficient.
Bilò et al.~\cite{BCFS21SingleSourceDSO_ESA} (with slightly different requirements on the set $B$)
were able to make do with three sets,
exploiting the undirectedness of the underlying graph.

We extend this to directed graphs using a hierarchical approach to find the critical paths.
Observe how the length parameter $L$ in \autoref{lemma:greedy-correctness} serves two roles.
The longer the paths, the longer it takes to compute $B$,
but the \emph{fewer} vertices suffice to intersect all paths.
Additionally, we have to compute the set of critical paths
which takes (at least) linear time in their length.
So $L$ has to fall just in the right range for the computation to be fast.
To achieve this, we use an exponentially growing sequence of lengths $L_1, L_2, \dots, L_{O(\log n)}$
and, instead of a single set, compute a sequence $B_1, B_2, \dots$ of exponentially shrinking sets
such that, in the $i$-th stage, $B_i$ hits, again for all $s,t \in V$ and $e \in E$,
some replacement path of length at least $L_i$.
However, this poses some new difficulties because now the collection of critical paths
has to be computed step by step.
Imagine in the $i$-th stage, we have already obtained the all the subsets $\mathcal{P}_j$, $j < i$, 
of paths with respective lengths $L_j$.
The key observation is that the hitting sets $B_j$ from the previous rounds
carry valuable information that can be harnessed to find the new set $\mathcal{P}_{i}$
faster, this then offsets the run time penalty 
incurred by the greater length of the new paths.
Our approach further relies on the existence of deterministic \emph{path-reporting}
distance sensitivity oracles that return the requested paths in constant time per edge/vertex.

\subparagraph{The HDPH Algorithm.}
We now describe the Hierarchical Double Pivots Hitting Sets (HDPH) algorithm that makes these ideas concrete.
It can be seen as a ``reference implementation'' of the framework.
For a specific application, one still has to adapt the details.
The algorithm is more general than what is needed for diameter oracles
in \Cref{theorem:large_diam_n^5/6,theorem:large_diam_upper_bound}.
For example, it also pertains to vertex failures.
Later, in \Cref{sec:derandomizing_existing,app:omissions_derad_existing},
we show an example how to modify the algorithm for other problems.

Let $C \ge 3/2$ be a constant.
The aim of the HDPH algorithm is to compute a sequence of sets $B_1, \dots, B_{\lceil \log_C n \rceil} \subseteq V$ of size $|B_i| = \Otilde(n/C^i)$ such that for all vertices
$s,t \in V$ and failure $f \in E \cup V$ with $d(s,t,f) \in (C^i, C^{i+1}]$ there exists a replacement path $P(s,t,f)$ that contains a pivot $z \in B_i$.
It assumes access to the ``APSP data'' of the original graph $G$, that is,
the distance $d(s,t)$ for all $s,t$ and a corresponding shortest path $P(s,t)$.
Also, it requires a deterministic path-reporting distance sensitivity oracle with constant query time (both for the distance and each reported edge) as a black box.

The HDPH algorithm is sketched in \Cref{alg:heirarchical-pivots}. In lines 1 and 2
it initializes $B_i = V$ for $i \le 2$. 
In lines 3-12, for $3 \le i \le \lceil \log_C n \rceil$, we iteratively compute the hitting sets $B_i$ by using the hitting sets from the previous $3$ iterations to obtain a set of shortest and replacement paths $\mathcal{P}_i$ of length $\Theta(n/C^i)$ that one needs to hit, and then use the greedy algorithm \textsf{GreedyPivotSelection} to compute the set of pivots $B_i$ which hits this set of paths $\mathcal{P}_i$. 
The paths are defined as follows. First, in line 4 we add to $\mathcal{P}_i$ shortest paths $P(x,y)$ whose length is in the range $(C^i, C^{i+1}]$ such that $x,y \in B_{i-3} \cup B_{i-2} \cup B_{i-1}$ are pivots from the last 3 iterations. Then, in lines 5-11, for every pair of pivots $x,y \in B_{i-3} \cup B_{i-2} \cup B_{i-1}$ whose shortest path $P(x,y)$ is of length at most $C^{i+1}$, and for every edge or every $f \in E(P(x,y)) \cup V(P(x,y))$ we query the DSO with $(x,y,f)$ to compute the distance $d(x,y,f)$. If $d(x,y,f) \in (C^{i-6}, C^{i+1}]$ then we use the DSO to also report a replacement path $P(x,y,f)$ and add it to $\mathcal{P}_i$.

\begin{algorithm}[t]
\label[alg]{alg:heirarchical-pivots}
\vspace*{.25em}
\caption{Hierarchical Double Pivots Hitting Sets (HDPH) Algorithm}
    \DontPrintSemicolon
    \KwIn{APSP data and a deterministic path-reporting DSO with $O(1)$ query time.}
    \KwOut{The hitting sets $B_0, \ldots, B_{\lceil \log_C n \rceil}$.}
    \For {$i \in [0, 2]$}
    {$B_i \gets V$ \;}
    \For {$i \in [3, \lceil \log_C n \rceil]$}
    { Let $\mathcal{P}_i = \{ P(x,y) \ | \ x,y \in B_{i-3} \cup B_{i-2} \cup B_{i-1} \text{ such that } d(x,y) \in (C^{i-6}, C^{i+1}]\}$ \; 
    \For {$x,y \in B_{i-3} \cup B_{i-2} \cup B_{i-1}$}
    {\If {$d(x,y) \le C^{i+1}$}
    {
    \For {$f \in E(P(x,y)) \cup V(P(x,y))$}
    {query the DSO for $d(x,y,f)$ \;
    \If {$d(x,y,f) \in (C^{i-6}, C^{i+1}]$}
    {query the DSO for $P(x,y,f)$ \;
    $\mathcal{P}_i \gets \mathcal{P}_i \cup \{P(x,y,f)\}$ \;
    }
    }
    }
    }
    $B_i \gets \textsf{GreedyPivotSelection}(\mathcal{P}_i)$ \;
    }
    \Return $B_0, \ldots, B_{\lceil \log_C n \rceil}$ \;
\vspace*{.25em}
\end{algorithm}

The next lemma proves the properties of the resulting hitting sets and the run time.

\begin{restatable}{lemma}{derandomizeFDO}
\label{lemma:derandomize-FDO}
Given the APSP data and a deterministic path-reporting DSO with $O(1)$ query time, the HDPH algorithm  deterministically computes, in $\Otilde(n^2)$ time, all the hitting sets $B_i$, with $0 \le i \le \lceil \log_{C} n\rceil$. For every
$0 \le i \le \lceil \log_C n\rceil$, it holds that $|B_i| = \Otilde(n / C^i)$. 
For every pair of vertices $s,t \in V$ and for every failing edge or vertex $f \in E \cup V$ such that $d(s,t,f) \in (C^i, C^{i+1}]$ there exists a pivot $z \in B_i$ such that $d(s,t,f) = d(s,z,f) + d(z,t,f)$. Finally, for every pair of vertices $s,t \in V$ such that $d(s,t) \in (C^i, C^{i+1}]$,
there exists a pivot $z \in B_i$ such that $d(s,t) = d(s,z) + d(z,t)$.
\end{restatable}

\begin{proof}
We first prove by induction that for every $i \in [0,\lceil \log_C n\rceil]$ it holds that $|B_i| = \Otilde(n / C^i)$. 
The claim trivially holds for $i \le 2$ as $B_0 = B_1 = B_2 = V$. 
For the inductive step, we assume that $|B_{j}| = \Otilde(n/C^{j})$ for every $j<i$. 
We show that the set of paths $\mathcal{P}_i$ contains $\Otilde(n^2/C^{i})$ paths, each of length $\Theta(C^i)$, and thus the result of the greedy algorithm $B_i \gets \textsf{GreedyPivotSelection}(\mathcal{P}_i)$ contains, by \autoref{lemma:greedy-correctness}, at most $\Otilde(n/C^i)$ vertices. Moreover, the runtime of the \textsf{GreedyPivotSelection} procedure is $\Otilde(n^2)$. 

For every $s,t \in V$, let $P(s,t)$ denote the shortest $s$-$t$-path in the APSP data.
There are two places where paths are added to $\mathcal{P}_i$.
In line 4, the algorithm adds shortest paths between vertices $x,y \in B_{i-3} \cup B_{i-2} \cup B_{i-1}$ whenever $d(x,y) = |P(x,y)| \in (C^{i-6}, C^{i+1}]$, and by the induction hypothesis there are $\Otilde(n^2/C^{2(i-1)}) = \Otilde(n^2/C^i)$ such pairs of vertices 
(since $|B_{j}| = \Otilde(n/C^{j})$ for every $j < i$).
Thus, the claim holds for the paths in line 4.
In line 11, the algorithm adds paths $P(x,y,f)$ to $\mathcal{P}_i$ only for pairs $x,y \in B_{i-3} \cup B_{i-2} \cup B_{i-1}$ and edges or vertices $f \in E(P(x,y)) \cup V(P(x,y))$ with $d(x,y) \le C^{i+1}$, there are $\Otilde(C^{i+1} \cdot (n^2/C^{2(i-1)})) = \Otilde(n^2/C^i)$ such triples $(x,y,f)$. The only paths added there are such that $d(x,y,f) \in (C^{i-6}, C^{i+1}]$ (due to the condition in line 9) and thus the length of $P(x,y,f)$ is $\Theta(C^i)$.
So the claim holds here as well.

Next, we prove that the runtime of the algorithm is $\Otilde(n^2)$.
We show that a single iteration of the for loop in lines 4-16 takes $\Otilde(n^2)$ time, and as there are $O(\log n)$ iterations for $i \in [3, \lceil \log_C n \rceil]$.
The number of pairs $x,y \in B_{i-3} \cup B_{i-2} \cup B_{i-1}$ is $\Otilde(n^2/C^{2(i-1)})$. The inner loop in lines 7-11 is executed only if $d(x,y) \le C^{i+1}$, therefore the number of edges $e \in P(x,y)$ is bounded by $C^{i+1}$ and hence the loop is executed at most $O(C^{i})$ times. Each iteration of this loop uses the black-box DSO to compute $d(x,y,f)$ in $O(1)$, and only if $d(x,y,f) \in (C^{i-6}, C^{i+1}]$ then we use the DSO to actually obtain the path $P(x,y,f)$ in $O(|P(x,y,f)|) = O(C^i)$ time.
This gives $\Otilde(n^2)$ for the second-most outer loop. 
We have already seen that computing $\textsf{GreedyPivotSelection}(\mathcal{P}_i)$ in line 12 takes $\Otilde(n^2)$ time as well.

We claim that for all $s,t \in V$ and every edge or vertex $f \in E \cup V$ such that $d(s,t,f) \in (C^i, C^{i+1}]$, there exists a pivot $z \in B_i$ such that $d(s,t,f) = d(s,z,f) + d(z,t,f)$.
That means, there is some $s$-$t$-replacement path that contains $z$.
This is clear for $i \le 2$.
Let $3 \le i \le \lceil \log_C n \rceil$ and suppose the claim holds for every $j < i$.  
Let $P(s,t,f) = (v_0 = s, v_1, \ldots, v_k = t)$ be an replacement path with
$k = d(s,t,f) \in (C^i, C^{i+1}]$.
We define the prefix
$P_1 = P(s,t,f)[s.. v_{\lceil k/C^3 \rceil}]$
and suffix
$P_2 = P(s,t,f)[v_{\lceil (1- 1/C^3)k \rceil} ..t]$.
Both subpaths have length in $(C^{i-3}, C^{i-2}]$.
It follows that there are pivots $x_1, x_2 \in B_{i-3}$ 
with $x_1 \in V(P_1)$, $x_2 \in V(P_2)$.
(Strictly speaking, we are merely guaranteed \emph{some} $s$-$v_{\lceil k/C^3 \rceil}$-replacement path that contains $x_1$, but we can choose $ P(s,t,f)$ so that its prefix is that path;
same with $P_2$.)

Let $P(x_1, x_2, f)$ be the replacement paths returned by the DSO on query $(x_1,x_2,f)$.
We claim that it is added to $\mathcal{P}_i$.
Observe that $d(x_1, x_2, f) \ge d(s,t,f) - |P_1| - |P_2| \ge (1-\frac{2}{C^2}) C^i > C^{i-6}$,
where we used the assumption $C \ge 3/2$ and thus $1-\frac{2}{C^2} > C^{-6}$.
Also, we have $d(x_1, x_2) \le d(x_1, x_2, f) \le d(s,t,f) \le C^{i+1}$.
If $d(x_1, x_2,f) = d(x_1,x_2)$, we may assume $P(x_1,x_2,f) = P(x_1,x_2)$, 
whence it was added in line 4.
Otherwise, the failure $f \in V(P(x_1, x_2))$ is on the path.
Since $x_1, x_2 \in B_{i-3} \cup B_{i-2} \cup B_{i-1}$,
$d(x_1, x_2) \le C^{i+1}$,   and
$d(x_1, x_2, f) \in (C^{i-6}, C^{i+1}]$ the path $P(x_1, x_2, f)$ is indeed added to $\mathcal{P}_i$ in line~11.
Due to $B_i \gets \textsf{GreedyPivotSelection}(\mathcal{P}_i)$,
there exists a vertex $z \in B_i$ such that
$z$ is on the path $P(x_1, x_2, f) \subseteq P(s,t,f)$ and thus
$d(s,t,f) = d(s,z,f) + d(z,t,f)$. 

The proof that for all $s,t \in V$ with $d(s,t) \in (C^i, C^{i+1}]$, there exists a pivot $z \in B_i$ such that $d(s,t) = d(s,z) + d(z,t)$ follows the same argument but is somewhat simpler
because the subpaths $P_1$ and $P_2$ are guaranteed to be added in line 4.
\end{proof}

\subparagraph{Derandomizing \Cref{theorem:large_diam_n^5/6,theorem:large_diam_upper_bound}.}
Recall that the oracle in \autoref{lemma:large_diam} has preprocessing time
$\Otilde(mn+n^4/(\varepsilon^3 \diam(G)^3))$.
For its derandomization, and that of
\Cref{theorem:large_diam_n^5/6,theorem:large_diam_upper_bound},
it is enough to choose $C = 2$,
compute APSP \emph{only} in the original graph $G$,
and preprocess the DSO of Bernstein and Karger\footnote{%
	Bernstein and Karger~\cite{BeKa09} derandomized their own DSO
	using a technique by King~\cite{King99FullyDynamicAPSP}.
}~\cite{BeKa09},
which takes $\Otilde(mn)$ time.
Let $i^*$ be the largest integer $i$ such that $2^{i} < L = \varepsilon \diam(G)/2$.
The set $B_{i^*}$ then hits,
for all $s,t \in V$ and $e \in E$ with $d(s,t,e) = \Theta(L)$, 
some replacement path $P(s,t,e)$, and it has the desired cardinality $\Otilde(n/L)$.

\section{Derandomizing Existing Sensitivity Oracles and Algorithms}
\label{sec:derandomizing_existing}

We now show how the HDPH algorithm can be adapted to derandomize existing sensitivity oracles.
In addition to our own technique, we also extensively use a recent breakthrough 
by Karthik and Parter~\cite{KarthikParter21DeterministicRPC}
in the derandomization of fault-tolerant algorithms.
We combine both tools and apply them to the distance sensitivity oracle of Ren~\cite{Ren22Improved}
and the SSRP algorithm of Chechik and Magen~\cite{ChMa19}
In the main part, we concentrate on the DSO because we think that it is a good illustration
of the combination of our work and that of Karthik and Parter~\cite{KarthikParter21DeterministicRPC}.
The treatment of the SSRP algorithm can be found in \autoref{app:omissions_derad_existing}.

\subsection{The Distance Sensitivity Oracle of Ren}
\label{subsec:derand_exist_DSO}

We start with the oracle of Ren~\cite{Ren22Improved}.
Recall that, for any two vertices $s,t \in V$ and edge $e \in E$, 
the replacement distance $d(s,t,e)$ is the length of a shortest $s$-$t$-path in $G-e$.
A distance sensitivity oracle (DSO) is a data structure 
that answers query $(s,t,e)$ with $d(s,t,e)$.
Ren~\cite{Ren22Improved} presented an algebraic DSO with a randomized preprocessing time
of $O(Mn^{2.7233})$ on graphs with positive integer edge weights in the range $[1,M]$
and $\Otilde(Mn^{(\omega + 3)/2}) = O(Mn^{2.6865})$ time on undirected graphs.
Notably, this was the first DSO with both constant query time and subcubic preprocessing, 
improving over previous work~\cite{AlonChechikCohen19CombinatorialRP,BeKa09,GrandoniVWilliamsFasterRPandDSO_journal,WY13}.
We derandomize it with a slight increase in running time
and obtain a deterministic DSO in time $O(Mn^{2.8068})$ on directed graphs
and $\Otilde(Mn^{(\omega + 6)/3}) = O(Mn^{2.7910})$ on undirected graphs.

The construction starts with a \textsf{Core} oracle
that only reports very small distances, 
this is then grown iteratively to cover longer paths
until the distance between all pairs of vertices are correctly determined.
More formally, for a positive real $r$, let an $r$\emph{-truncated DSO}
report, upon query $(s,t,e)$, the value $d(s,t,e)$
if it is at most $r$, and $+\infty$ otherwise.
The \textsf{Core} is an \mbox{$n^{\alpha}$-truncated} DSO
for some carefully chosen exponent $\alpha \in (0,1)$. 
Each iteration invokes the procedure \textsf{Extend} to turn an \mbox{$r$-truncated} DSO
into an $(3/2) \nwspace r$-truncated DSO.
Note that we can assume $M = \Otilde(n^{(3-\omega)/2})$ 
as otherwise the deterministic oracle in~\cite{BeKa09} with an $\Otilde(mn)$ preprocessing time
already achieves $\Otilde(Mn^{(\omega + 3)/2})$, even on directed weighted graphs.
Hence, $\log_{3/2}(Mn) = O(\log n)$ rounds of growing suffice to built the full oracle.

The iterative approach has the advantage that $r$-truncated DSOs for small $r$
can be computed fast.
A bridging-set idea, see~\cite{Zwick02DirectedAPSP}, is used for the extension.
This significantly increases the query time as the oracle has to cycle through the 
whole bridging set to compute the distance.
Ren~\cite{Ren22Improved} uses a clever observation, there attributed to Bernstein and Karger~\cite{BeKa09},
to reduce the query time of the extended DSO back to a constant,
called the \textsf{Fast} procedure.

Randomness is employed at two points.
First, the \textsf{Core} uses a series random subgraphs of $G$.
Secondly, \textsf{Extend} randomly samples a set of pivots to hit all replacement paths
of length between $r$ and $(3/2) \nwspace r$.
The subsequent reduction in query time is deterministic.\footnote{
	The relevant \cite[Section~3]{Ren22Improved} is phrased as randomized,
	but based on the derandomizable algorithm in~\cite{BeKa09}.
}
The \textsf{Core} can be derandomized using a recent result by Karthik and Parter~\cite{KarthikParter21DeterministicRPC}.
To derandomize \textsf{Extend}, we adapt our technique introduced above.
The key differences are that we now have to take care of the edge weights,
that is, the number of vertices of a path may be much smaller than its length.
Also, due to the iterative approach of not only the derandomization
but the actual construction via truncated DSOs, we cannot assume 
to have access to all relevant paths right from the beginning.
Instead, we have to make sure that all intermediary oracles are \emph{path-reporting}
and that for the construction of the current hitting set we only use paths of length at most $r$.
The deterministic \textsf{Core} oracle hinges on the following lemma.

\begin{lemma}[Karthik and Parter~\cite{KarthikParter21DeterministicRPC}]
\label{lem:KarthikParter}
	Given a (possibly weighted) graph $G$ on $n$ vertices and a positive real $r = n^\alpha$
	for some $\alpha \in (0,1)$,
	there is a deterministic algorithm computing $k = O(r^2)$
	spanning subgraphs $G_1, \dots, G_k$ of $G$ in time $\Otilde(kn^2)$ 
	such that for any pair of vertices $s,t \in V$, edge $e \in E$,
	and replacement path $P(s,t,e)$ on at most $r$ edges, there exists an index $i$ such that $G_i$
	does not contain the edge $e$ but all edges of $P(s,t,e)$.
\end{lemma}

\noindent
This derandomizes a construction by Weimann and Yuster~\cite{WY13} with the crucial
difference that the latter is only required to produce subgraphs such that
for all pairs of vertices $s,t$ and edges $e$ that admit possibly multiple replacement paths on at most $r$ edges \emph{at least one} (instead of all) of them survives in one of the graphs $G_i$ in which $e$ was removed.
This relaxed condition is actually enough to build an $r$-truncated DSO and allows one to make do with only $\Otilde(r)$ random subgraphs, while we have $O(r^2)$ deterministic graphs.
See also the discussion in Section 1.3 of~\cite{KarthikParter21DeterministicRPC}.
This is the sole reason for the increased running time 
compared to the original result of Ren~\cite{Ren22Improved}.
 
Given a graph $G$ with integer edge weights in the range $[1,M]$,
we invoke \autoref{lem:KarthikParter} to obtain the subgraphs $G_i$.
Recall that $r = n^{\alpha}$ and let $\omega(1\,{-}\,\alpha)$ be the infimum over all $w$ such that
rectangular integer matrices with dimensions $n \times n^{1-\alpha}$ and $n^{1-\alpha} \times n$
can be multiplied using $O(n^w)$ arithmetic operations,
$\omega = \omega(1)$ is the usual square matrix multiplication coefficient.
Using a variant of Zwick's algorithm~\cite{Zwick02DirectedAPSP},\footnote{%
	The algorithm in~\cite{Zwick02DirectedAPSP} is also phrased as randomized,
	in the same work it is explained how to derandomize it,
	increasing the running time only by polylogarithmic factors.
	The same holds for~\cite{ShZw99}.
}
we compute APSP restricted to paths 
on at most $r$ edges in time $\Otilde(M n^{\omega(1-\alpha)}r)$ per subgraph.
If $G$ is undirected, then this can be done faster,
namely, in $\Otilde(Mn^\omega)$ per graph with the algorithm of Shoshan and Zwick~\cite{ShZw99}.
Both algorithms in~\cite{ShZw99,Zwick02DirectedAPSP} can be adjusted to also compute
the actual paths, represented as predecessor trees, which increases the running time
only by logarithmic factors.

To answer a query $(s,t,e)$ we cycle through 
the $G_i$ and, in case the edge $e$ is missing,
retrieve the distance $d_{G_i}(s,t)$.
By \autoref{lem:KarthikParter}, the minimum over all retrieved distances
is the correct replacement distance $d(s,t,e)$.
If this minimum is larger than $r$ or no distance has been retrieved at all 
(as the paths take more than $r$ edges), we return $+\infty$.
Since the edge weights are positive integers, every path of \emph{length} at most $r$
uses at most $r$ edges, so we indeed obtain an $r$-truncated DSO.
If an actual replacement path is requested, we return a shortest $s$-$t$-path
in one of the $G_i$ that attain the minimum.
The resulting oracle has query time $\Otilde(r^2)$ and a
$\Otilde(n^2r^2 + Mn^{\omega(1-\alpha)}r^3) = \Otilde(Mn^{\omega(1-\alpha)}r^3)$
preprocessing time on directed graphs
(using $\omega(1-\alpha) \ge 2$). 
On undirected graphs, this improves to
$\Otilde(n^2 r^2 + Mn^{\omega}r^2) = \Otilde(Mn^{\omega}r^2)$.

As a technical subtlety, 
the \textsf{Fast} procedure needed to reduce the query time
requires \emph{unique} shortest paths\footnote{%
	By \emph{unique shortest paths}, we mean a collection $\mathcal{P}$
	containing one shortest path for each pair of vertices such that,
	for all $s,t \in V$, if $P(s,t) \in \mathcal{P}$ is the shortest path from $s$ to $t$,
	then for every vertex $u$ on $P(s,t)$, 
	the path $P(s,u) \in \mathcal{P}$ is the prefix $P(s,t)[s..u]$
	and $P(u,t) \in \mathcal{P}$ is the suffix $P(s,t)[u..t]$.
}
of the original graph $G$.
They can be computed in time $\Otilde(M^{1/2} \nwspace n^{(\omega+3)/2})$~\cite{DP09MaxMin,Ren22Improved}.
We will see later that this is not the bottleneck of the preprocessing.

\begin{lemma}[Ren~\cite{Ren22Improved}]
\label{lem:Rens_observation}
	From a directed graph $G$ with integer edge weights in $[1,M]$,
	unique shortest paths, and an $r$-truncated DSO with preprocessing time $P$
	and query time $Q$, one can built in deterministic time $P + \Otilde(n^2) \cdot Q$
	a $r$-truncated DSO for $G$ with $O(1)$ query time.
\end{lemma}

\noindent
Without access to unique paths, the running time increases to 
$P + \Otilde(Mn^2) \cdot Q$, see~\cite{Ren20}.
If the oracle with query time $Q$ (for the distance)
is path-reporting (in $O(1)$ time per edge),
then the new oracle is path-reporting with $O(1)$ query time
(for distances and edges)~\cite{Ren22Improved}.

We now turn to the main part, where we derandomize the \textsf{Extend} procedure
that turns an \mbox{$r$-truncated} DSOs into $(3/2) \nwspace r$-truncated DSOs.
We adapt our technique to the iterative manner of construction
and to the integer weights on the edges.
In each stage, we only have access to a truncated DSO.
Still, we show how to deterministically compute
a sequence $B_1, B_2, \dots$ of smaller and smaller sets,
where $B_i$ is used to derandomize the $i$-th application of \textsf{Extend}.
Again, the construction of $B_i$ depends on the previous sets of pivots, namely, on $B_{i-2}$. 
We first describe how to obtain the $B_i$ satisfying certain useful properties
and afterwards verify that these properties indeed suffice to make \textsf{Extend}
deterministic.

\begin{lemma}
\label{lem:hierarchical_pivots}
	Let $r_1 \ge 1$ be a real number and define $r_{i+1} = (3/2) \nwspace r_{i}$.
	For a graph $G$ with integer edge weights in $[1,M]$,
	let $\{B_i\}_{i \ge 1}$ be a family of subsets of $V$,
	such that, for each $i$,
	\begin{inparaenum}
		\item[\textsf{\emph{(i)}}] $|B_i| =  \Otilde(Mn/r_i)$;
		\item[\textsf{\emph{(ii)}}] for every pair of $s,t \in V$ and $e \in E$ with $r_i/2 -M \le d(s,t,e) \le r_i$,
			there exists a replacement path $P(s,t,e)$ that contains a vertex of $B_i$.
	\end{inparaenum}
	With access to the shortest paths, 
	a path-reporting $r_i$-truncated DSO with $O(1)$ query time, and the sets $B_j$ with $j < i$,
	one can compute each $B_i$ deterministically in time $\Otilde(M^2 \nwspace n^2 + n^2 r_1^2)$.
\end{lemma}

\begin{proof}
	The proof is by induction over $i$.
	For the construction of $B_i$, we use the previous set $B_{i-2}$.
	We set $B_{-1} = B_{0} = V$ to unify notation.
	Following the outline of the derandomization technique,
	we first assemble a set $\mathcal{P}$ of paths 
	and then greedily compute a hitting set.
	
	For each pair of vertices $x,y \in B_{i-2}$, we check whether
	the $x$-$y$-distance in the base graph $G$ is at most $r_i$
	and, if so, retrieve a shortest path $P(x,y)$.
	If $P(x,y)$ additionally has length at least $r_i/18$, we add it to $\mathcal{P}$.
	For each edge $e$ on $P(x,y)$ (regardless of the path being added to $\mathcal{P}$),
	we query the $r_i$-truncated DSO whether
	the replacement distance is  $d(x,y,e) \in [\frac{r_i}{18},r_i]$.
	If so, we request a corresponding replacement path $P(x,y,e)$
	to add it to $\mathcal{P}$.

	Due to the positive weights, those paths have at most $r_i$ edges 
	and can be obtained in time $O(r_i)$.	
	Assembling $\mathcal{P}$ thus takes time $O(|B_{i-2}|^2 \nwspace r_i^2)$.
	If $i \le 2$, this is $O(n^2 \nwspace r_1^2)$ since $r_2 = (3/2) \nwspace r_1$.
	For $i \ge 3$, we get
	$\Otilde((M n/r_{i-2})^2 \cdot r_i^2) = \Otilde(M^2 \nwspace n^2)$ instead.

	We deterministically compute a hitting set $B_i$ for $\mathcal{P}$.	
	Since $\mathcal{P}$ contains at most $|B_{i-2}|^2 \cdot r_i$ paths
	with at least $r_i/(18M)$ edges each, whence $\Omega(r_i/M)$ vertices,
	the set $B_i$ has $\Otilde(n/(r_i/M)) = \Otilde(Mn/r_i)$ vertices
	and is computable in time $\Otilde(|\mathcal{P}| \cdot (r_i/M))$.
	As before, for $i \le 2$, this is $\Otilde(n^2 r_1^2/M)$;
	and $\Otilde(M n^2)$ otherwise.
	We get a running time of $\Otilde(M^2 \nwspace n^2 + n^2 r_1^2)$.
	
	It is left to prove that $B_i$ indeed hits at least one replacement path
	for all $s,t \in V$ and $e \in E$ that satisfy 
	$d(s,t,e) \in [\tfrac{r_i}{2} - M, \nwspace r_i]$.
	Let $P(s,t,e)$ be such a path
	and define $u$ to be the first vertex on $P(s,t,e)$ (starting from $s$) 
	such that $d(s,u,e) \ge (2/9) \nwspace r_i - M$.
	If $i \ge 3$, then $r_{i-2} = (4/9) \nwspace r_i$,
	whence $d(s,u,e) \in [\frac{r_{i-2}}{2} -M, \nwspace \frac{r_{i-2}}{2})$
	and the induction hypothesis implies that there is \emph{some} replacement path $P'$
	from $s$ to $u$ avoiding the edge $e$ such that $B_{i-2}$ contains one vertex of $P'$.
	Otherwise, if $i \le 2$, the same fact simply follows from $B_{i-2} = V$.

	The path $P'$ is not necessarily equal to the prefix of $P(s,t,e)[s..u]$
	(but they have the same length $d(s,u,e)$).
	Replacing $P(s,t,e)[s..u]$ with $P'$ gives a new replacement path
	that now has a pivot $x \in B_{i-2}$ on its prefix.
	Slightly abusing notation, we use $P(s,t,e)$ to denote also the updated path.
	Let $v$ be the last vertex on $P(s,t,e)$ with $d(v,t,e) \ge (2/9) \nwspace r_i - M$.
	By the same argument, we can assume that the suffix 
	of $P(s,t,e)[v..t]$ contains a pivot $y \in B_{i-2}$.
	In the remainder, we show 
	that there is some replacement path $P(x,y,e)$ that is hit by a vertex in $B_i$.
	If so, replacing the middle part $P(s,t,e)[x..y]$ with $P(x,y,e)$
	finally proves the existence of a replacement path from $s$ to $t$ avoiding $e$
	and containing a vertex of $B_i$.
	
	By the choice of the pivots $x,y$ and the assumption $d(s,t,e) \in [\frac{r_i}{2}-M, r_i]$,
	the replacement distance $d(x,y,e)$ satisfies
	$$r_i \ge d(s,t,e) \ge d(x,y,e) \ge d(s,t,e) - d(s,u,e) - d(v,t,e) \ge
	    d(s,t,e) - 2 \left(\frac{2 \nwspace r_i}{9} - M\right) \ge \frac{r_i}{18}.$$
	First, assume that the shortest path $P(x,y)$ in the base graph $G$ 
	does not contain the edge $e$.
	Then, $P(x,y)$ can serve as the replacement path.
	It has length $d(x,y) = d(x,y,e)$ between $r_i/18$ and $r_i$,
	and we added it to $\mathcal{P}$.
	Otherwise, it holds that $e \in P(x,y)$.
	Observe that $d(x,y) \le d(x,y,e) \le r_i$ remains true.
	Therefore, we have queried the $r_i$-truncated DSO with the triple $(x,y,e)$.
	Due to $d(x,y,e) \ge r_i/18$, 
	we received a replacement path $P(x,y,e)$, which we added to $\mathcal{P}$.
	In both cases, some replacement path is hit by $B_i$, as desired.
\end{proof}

At first glance, it looks like the quadratic dependence on $M$ is too high
to be used in the derandomization.
However, recall that we can assume $M = \Otilde(n^{(3-\omega)/2})$.
Over the $O(\log n)$ rounds with $i \ge 3$ and with access to the appropriately truncated DSOs,
we can compute the sets $B_3, B_4, \dots$ in
total time $\Otilde(M^2 \nwspace n^2) = \Otilde(M n^{(7-\omega)/2}) = \Otilde(M n^{2.5})$
even if $\omega = 2$.

The next lemma is the last tool we need to construct the deterministic DSO.

\begin{lemma}
\label{lem:bridging_set}
	Let $G$ be a graph with integer edge weights in the range $[1,M]$, 
	$r \ge 1$ a real number,
	and $B \subseteq V$ a set of $\Otilde(Mn/r)$ vertices such that 
	for every pair of $s,t \in V$ and $e \in E$ with $r/2 -M \le d(s,t,e) \le r$,
	there exists a replacement path $P(s,t,e)$ that contains a vertex of $B$.
	Given an $r$-truncated DSO for $G$ with $O(1)$ query time and the set $B$,
	one can, without further preprocessing,
	construct an $(3/2) \nwspace r$-truncated DSO with query time $\Otilde(Mn/r)$.
	Moreover, if the $r$-truncated DSO is path-reporting, so is the $(3/2) \nwspace r$-truncated one.
\end{lemma}

\begin{proof}
	For any query $(s,t,e)$, let $D(s,t,e)$ denote the 
	returned value by the $r$-truncated DSO.
	If $D(s,t,e) \neq +\infty$, we also take this as the answer of the 
	\mbox{$(3/2) \nwspace r$-truncated DSO}.
	Otherwise, define $\ell = \min_{z \in B} \{ D(s,z,e) + D(z,t,e) \}$.
	If $\ell \le (3/2) \nwspace r$, we return $\ell$, and $+\infty$ else.
	Path queries are handled in the same fashion.
	In the case of $D(s,t,e) \neq +\infty$, we pass on the path $P(s,t,e)$
	returned by the $r$-truncated DSO.
	If $\ell \le (3/2) \nwspace r$, we return the concatenation of $P(s,z,e)$ and $P(z,t,e)$
	for some pivot $z \in B$ that attains the minimum $\ell$.
	The query time is $O(|B|) = \Otilde(Mn/r)$ for the distance, 
	after which the path can be returned in $O(1)$ per edge.
	
	It is clear that the query algorithm is correct whenever $d(s,t,e) \le r$
	as those queries are entirely handled by the given truncated DSO.
	Moreover, even if $d(s,t,e) > r$, then $\ell$ is an upper bound for $d(s,t,e)$
	because all sums $D(s,z,e) + D(z,t,e)$ correspond to some
	path from $s$ to $t$ avoiding $e$, but not necessarily a shortest path.
	
	Let $P = P(s,t,e)$ be a replacement path of length between $r$ and $(3/2) \nwspace r$,
	$u$ the first vertex on $P$ (seen from $s$) with $d(u,t,e) \le r$,
	and $v$ the last vertex on $P$ with $d(s,v,e) \le r$.
	Note that $v$ lies between $u$ and $t$ on the path,
	whence  $d(u,v,e) \le r$.
	We further have
	\begin{equation*}
		d(u,v,e) \ge d(s,t,e) - d(s,u,e) - d(v,t,e) = d(u,t,e) + d(s,v,e) - d(s,t,e) 
		\ge 2r - \frac{3}{2}r = \frac{r}{2}.
	\end{equation*}
	
	By the properties of $B$, there exists some replacement path from $u$ to $v$ avoiding $e$
	that contains a pivot $z \in B$.
	With the usual argument of swapping parts of the path, we can assume $z$ lies on the middle
	section of $P(s,t,e)$ between $u$ and $v$.
	By construction, we have $\max\{d(s,z,e), d(z,t,e)\} \le r$ so both distances 
	(and corresponding paths) are correctly determined by the $r$-truncated DSO.
	In summary, we get $\ell \le d(s,z,e) + d(z,t,e) = d(s,t,e)$
	and the returned value $\ell$ is indeed the correct replacement distance.
\end{proof}

We are left to prove the final running time of the construction.
Let $r = n^\alpha$ be the cut-off point for the distances at which we start the iterative growing.
We build the \textsf{Core} DSO using the $O(r^2)$ subgraphs, compute unique shortest paths in $G$,
followed by $O(\log n)$ iterations of \textsf{Extend} and \textsf{Fast} invocations,
including the computation of the $B_i$.
First, suppose the graph $G$ is undirected.
The total time is then
\begin{gather*}
	\Otilde(Mn^\omega r^2) + \Otilde(M^{1/2} n^{(\omega + 3)/2}) + \Otilde(M n^{2.5} + n^2 r^2) 
		+ \Otilde(n^2) \cdot \sum_{i=1}^{O(\log n)} \Otilde\!\left( \frac{Mn}{(3/2)^i \nwspace r} \right)\\
		= \Otilde\!\left(\! M^{1/2} \nwspace n^{(\omega + 3)/2} + M n^{2.5} + Mn^\omega r^2 + \frac{Mn^3}{r}\right)\\
		= \Otilde\!\left(M^{1/2} \nwspace n^{(\omega + 3)/2} + M n^{\max\{2.5, \ \omega+2\alpha, \ 3-\alpha \}} \right).
\end{gather*}

\noindent
This is minimum for $\alpha = 1 - (\omega/3)$, where we get a running time of $\Otilde(Mn^{(\omega+6)/3})$.

For directed graphs, determining the best $\alpha$ is a bit more involved.
Recall that $O(n^{\omega(1-\alpha)})$ is the time needed to multiply 
a $n \times n^{1-\alpha}$ matrix with an $n^{1-\alpha} \times n$ matrix.
Computing the \textsf{Core} oracle takes time 
$\Otilde(M n^{\omega(1-\alpha)} \nwspace r^3) = \Otilde(M n^{\omega(1-\alpha)+3\alpha})$.
With a similar calculation as above, we obtain a total preprocessing time of
$\Otilde\!\left(M^{1/2} \nwspace n^{(\omega + 3)/2} + M n^{\max\{2.5, \ \omega(1-\alpha)+3\alpha, \ 3-\alpha \}} \right)$.
This is minimized if $\alpha$ solves the equation $\omega(1\,{-}\,\alpha) = 3-4\alpha$.
Le Gall and Urrutia~\cite{LeGall18RectangularMatrixMultiplication}
gave the current-best estimates for the values of the function $\omega$.
This shows that $1\,{-}\,\alpha$ is between $0.8$ and $0.85$,
and we have $\omega(0.8) \le 2.222256$ as well as $\omega(0.85) \le 2.258317$.
We exploit the fact that $\omega$ is convex~\cite{LottiRomani83RectangularMatrixMultiplication},
giving
\begin{equation*}
	\omega(1-\alpha) \le \frac{(\alpha-0.15) \nwspace \omega(0.8) + (0.2-\alpha) \nwspace \omega(0.85)}{0.85-0.8}
		\le 2.3665 - 0.72122 \nwspace \alpha
\end{equation*}

\noindent
Equating the latter with $3-4\alpha$ yields the estimate $\alpha \le 0.193212$,
which in turn implies a preprocessing time of $\Otilde(Mn^{2.806788})$.

\bibliographystyle{plainurl} 
\bibliography{merger_bib}

\appendix

\section{Omitted Parts of \autoref{sec:FDOs_digraphs}}
\label{app:omitted_proofs}

This appendix contains the omissions that had to be made in \autoref{sec:FDOs_digraphs},
namely the proofs of the space lower bounds for FDOs on directed graphs
as well as the discussion of the eccentricity oracle for multiple failures on DAGs.

\subsection{Space Lower Bounds}
\label{subapp:proofs_omitted_Sec_3}

\lemmalowerbounddigraph*

\begin{proof}
	We encode certain binary matrices in the fault-tolerant diameters 
	of a directed graph $G$, that is, in the values $\diam(G{-}e)$ for edges $e \in E$.
	To simplify the writing, we assume $\sqrt{m}$ is an integer;
	otherwise, every occurrence $\sqrt{m}$ below can be replaced by $\lfloor \sqrt{m} \rfloor$
	without affecting the result.
	Let $X$ be a binary $\sqrt{m} \times \sqrt{m}$ matrix such that for all indices $1 \le i, j \le \sqrt{m}$,
	we have $X[i,1] = X[1,j] = 1$, meaning that the first row and the first column
	are the all-ones vector.
	The class $\G$ contains a graph for each such matrix.
	
	We now describe the construction of $G = G(X)$.
	Assume for now that $D$ is odd, in the end we show the easy adaption for even $D$.
	Define $t=(D{-}1)/2$; by our assumptions, we have $t \ge 1$.
	The vertex set $V$ is partitioned into five sets $V_A,V_B,V_C,V_D,V_R$ 
	with respective cardinalities $t\sqrt{m}$, $\sqrt{m}$, $\sqrt{m}$, $t\sqrt{m}$,
	and $n- (D{+}1)\sqrt{m}$.
	The graph has a total of $(2t{+}2)\sqrt{m} + (n-(D{+}1)\sqrt{m}) = n$ vertices
	(where we used $t = (D{-}1)/2$ and $D < n/\sqrt{m}$).
	
	\begin{itemize}
		\item The vertices in $V_A$ (resp. $V_D$) are denoted by $a_{k,i}$ (resp.  $d_{k,i}$) 
			for all pairs of indices with $1\leq k\leq t$ and $1\leq i\leq \sqrt{m}$. 
		\item The vertices in $V_B$ (resp. $V_C$) are denoted by $b_i$ (resp. $c_i$) 
			for $1\leq i \leq \sqrt{m}$.
		\item The remaining vertices in $V_R$ are denoted by 
			$r_\ell$ for $1 \le \ell \le n-(D{+}1)\sqrt{m}$. 
	\end{itemize}

	\noindent
 	Those vertices are joined by the following edges. 
 	\autoref{fig:lower_bound_digraphs_extensive} provides an overview.
 	
	\begin{itemize}
		\item For each $1\leq i,j\leq N$, the edges $(a_{t,i},c_j)$ and $(b_i, d_{1,j})$ 
			(red edges in \autoref{fig:lower_bound_digraphs_extensive})
			are present in $G$ if and only if $X[i,j]=1$.
			Note that by our assumption on $X$, for each $i,j$, 
			$(a_{t,i}, c_1)$, $(b_i, d_{1,1})$, $(a_{t,1}, c_j)$, 
			and $(b_1, d_{1,j})$ are \emph{always} edges of $G$.
		\item Each vertex $a_{k,i}$, has an out-edge, namely, $(a_{k,i},a_{k+1,i})$ for $k<t$,
			or $(a_{t,i},b_i)$ (if $k=t$).
		\item Each $d_{k,j}$ has an in-edge $(d_{k-1,j},d_{k,j})$ for $k>1$, or
			$(c_j,d_{1,j})$.
		\item From set $V_B$ to $V_C$, graph $G$ has a complete bipartite subgraph,
			that is, $(b_i,c_j)$ is an edge for all $1 \le i,j \le \sqrt{m}$.
		\item Each $x \in V_C\cup V_D$ is connected to each one in 
			$\{a_{1,1}, \dots, a_{1,\sqrt{m}}\} \subseteq V_A$.
		\item Each $r_\ell \in V_R$ participates in the three edges
			$(r_{\ell},c_1)$, $(c_1,r_\ell)$, $(r_{\ell},a_{1,1})$. 
	\end{itemize}

	\begin{figure}[t]
	\centering
	\includegraphics[scale=.88]{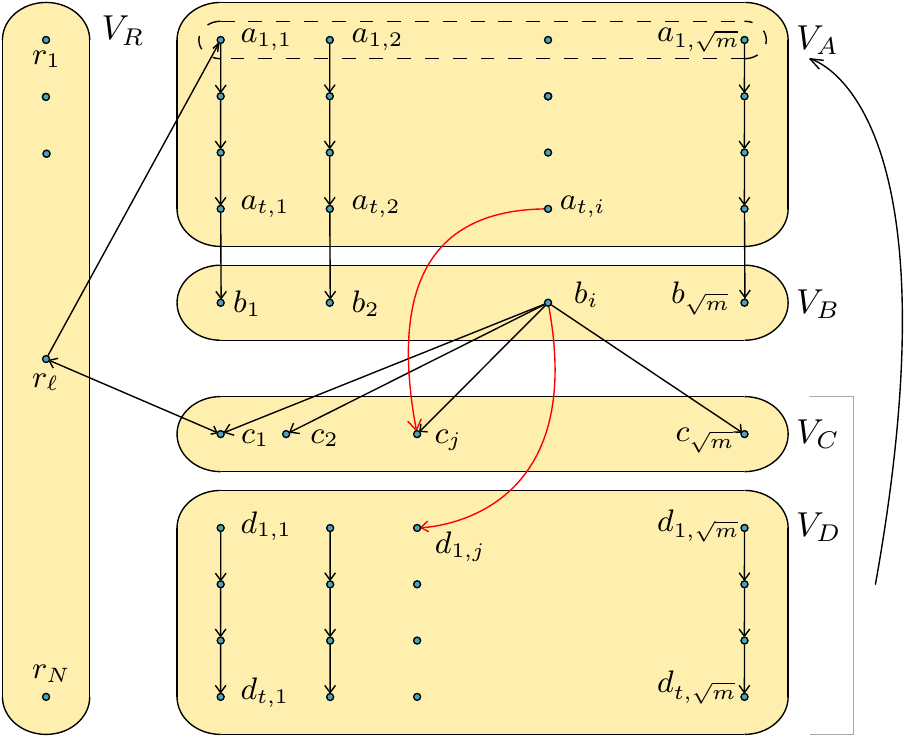}
	\caption{An overview of the space lower bound for graphs with diameter $D$ (\autoref{lemma:lower_bound_digraph}).
		We have $t = (D\,{-}\,1)/2$ and $N = n - (D\,{+}\,1) \sqrt{m}$
		to ensure $n$ vertices and $\Theta(m)$ edges.
		The red edges $(a_{t,i}, c_j)$ and $(b_i,d_{1,j})$ are present if and only if 
		the matrix entry $X[i,j] = 1$.
		If the edge $(b_i,c_j)$ fails, the red edges keep the diameter at $D$,
		in their absence the diameter increases to $3(D \,{-}\, 1)/2$.}
	\label{fig:lower_bound_digraphs_extensive}
	\end{figure}

	It is easy to see that $G$ has $\Theta(n + (\sqrt{m})^2) = \Theta(m)$ edges.
	We argue that its diameter is $2t + 1 = D$.
	First consider the subgraph induced by $V{\setminus}V_R$.
	It has $2t+2$ layers of vertices connected in a circular manner,
	so we can get from $V_A \cup V_B$ to any vertex in $V_C \cup V_D$
	using a path of length at most $2t+1$; 
	the opposite direction is even shorter.
	Moving between any two vertices within $V_A \cup V_B$ can be done by first reaching a vertex in
	$V_C$ in at most $t$ steps, with the last step being $(a_{t,i}, c_1)$ or $(b_i,c_j)$ (for any $j$).
	After that, one can reenter $V_A \cup V_B$ and reach the target
	in at most $t+1$ additional steps;
	the reasoning for moving within $V_C \cup V_D$ is similar.
	Hence, the diameter of the subgraph $G[V{\setminus}V_R]$ is at most $2t+1$.
	Any vertex $r_\ell \in V_R$ is connected to the two ``hub nodes'' $c_1$ and $a_{1,1}$
	allowing it to reach all other vertices fast, 
	while at the same time preventing $r_\ell$ to be used as a shortcut.
	In more detail, any other $r_{\ell'}$, with $\ell' \neq \ell$, can be reached via the path 
	$(r_\ell, a_{1,1}, \dots a_{t,1}, c_1, r_{\ell'})$ in $t+2 \le 2t+1$ steps.
	A very similar route $(r_\ell, \dots a_{t,1}, c_j, \dots, d_{k,j})$ of length at most $2t+1$
	leads to any vertex in $V_C \cup V_D$.
	Finally, the vertices in $V_A \cup V_B$ are reached
	using a path of the form $(r_\ell, c_1, a_{1,i}, \dots, a_{t,i}, b_i)$ of length at most $t+2$.
	The same arguments show that any vertex in $V{\setminus}V_R$ has distance at most
	$t+1$ to $c_1$, one more edge gives the path to any $r_\ell$.

	We have seen $\diam(G) \le D = 2t+1$,
	the fact that the diameter is also not smaller than $D$ 
	is witnessed by the distance $d(a_{1,i},b_j)$
	for any $i \neq j$ (using $m \ge 4$ and thus $\sqrt{m} \ge 2$).
	One can reach $b_j$ only via $a_{1,j}$, taking $t$ steps,
	and a shortest path from $a_{1,i}$ to $a_{1,j}$
	is $(a_{1,i}, \dots, a_{t,i}, c_{1}, a_{1,j})$
	of length $t+1$.
	 
	Now consider the edge $e_{ij}=(b_i,c_j)$ for $j \neq 1$.
	Suppose first that $X[i,j] = 1$ holds, whence the edges
	$(a_{t,i},c_j)$ and $(b_i, d_{1,j})$ are present.
	The only shortest path in $G$ that crucially depends on $e_{ij}$ is $(b_i,c_j)$ itself. 
	To see this, let $P$ be any other path that uses $e_{ij}$.
	It either also goes through vertex $a_{t,i}$
	(resp.\ $d_{1,j}$)  and the subpath $P[a_{t,i} .. c_j]$ (resp.\ $P[b_i .. d_{1,j}]$)
	can be shortened to the single edge $(a_{t,i},c_j)$ (to edge $(b_i, d_{1,j})$), 
	so $P$ is not a shortest path;
	or it has $(b_i, c_j, a_{1,\ell})$ as a subpath for some $1 \le \ell \le \sqrt{m}$, 
	which can be substituted by $(b_i,c_{j'},a_{1,\ell})$ for any $j' \neq j$. 
	In the graph $G{-}e_{ij}$, the distance from $b_i$ to $c_j$ is at most $t+2 \le D$
	using the path $(b_i, d_{1,j}, a_{1,1}, \dots, a_{t,1}, c_j)$.
	Both things together show that failing the edge $e_{ij}$
	leaves the fault-tolerant diameter at $D$ in case $X[i,j] = 1$.

	We show that $\diam(G- e_{ij})$ is at least $3t+1 = (3D{-}1)/2$ if $X[i,j]=0$,
	that is, if neither $(a_{t,i},c_j)$ nor $(b_i, d_{1,j})$ are present.
	Indeed, we then have $d(a_{1,i},d_{t,j},e_{ij})\geq 3t+1$.
	The only way to reach $d_{t,j}$ in $G-e_{ij}$ is via vertex $c_j$ 
	for which we have to go through $a_{t,i'}$ and therefore through $a_{1,i'}$
	for some $i'$ with $X[i',j] = 1$ (that is, $i' \neq i$;
	we can safely assume $i' = 1$).
	It requires $2t$ steps from $a_{1,1}$ to get to $d_{t,j}$,
	and the $a_{1,i}$-$a_{1,1}$-path has length $t+1$ as above.

	Any data structure that can distinguish whether the diameter stays at $D$
	or rises to $(3D{-}1)/2$ must differ by at least $1$ bit for any two graphs in $\G$.
	Since there are $2^{(\sqrt{m}-1)^2}$ admissible matrices $X$/graphs in $\G$,
	the $\Omega(m)$ bound follows.
	
	Only small changes are needed to make this work for even values of $D$.
	Set $t = (D/2) - 1$.
	Instead of connecting $V_C \cup V_D$ directly with each vertex in 
	$\{a_{1,1}, \dots, a_{1,\sqrt{m}}\}$, we introduce an intermediate vertex $v$
	(taken from the reservoir $V_R$ to keep the total number of vertices at $n$).
	We add the edges $(x,v)$ and $(v,a_{1,i})$ for each $x \in V_C \cup V_D$
	and $1 \le i \le \sqrt{m}$.
	The same arguments as above, replacing subpaths of the form $(c_1, a_{1,j})$
	by $(c_1,v,a_{1,j})$ show that the diameter of $G$ is $2t+2 = D$,
	remains there if the edge $e_{ij}$ fails but $(a_{t,i},c_j)$ and $(b_i,d_{1,j})$
	are present, and rises to $3t+2 = (3D/2)-1$ otherwise.
\end{proof}

\begin{figure}[t]
\centering
\includegraphics[scale=.52]{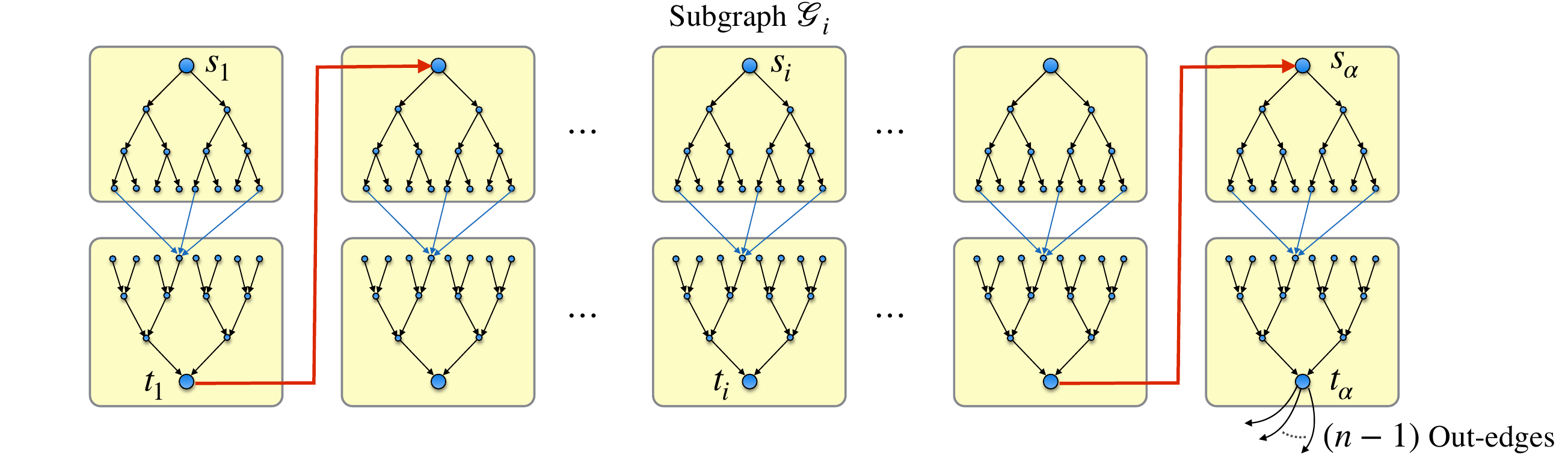}
\caption{Illustration of space lower bound for fault-tolerant diameter oracles 
	with finite stretch (\autoref{lemma:connectivity_lower_bound}).
	All binary trees have height $f$ and are directed downwards.
	For the upper trees, this means away from the root $s_i$,
	and towards the root $t_i$ in the lower ones.
	Edges between the subtrees in $G_i$ depend on the binary matrix $X_i$.
	The out-edges of $t_\alpha$ wrap around to \emph{all} other vertices. 
}
\label{figure:diam_fin_str}
\end{figure}

\connectivitylowerbound*

\begin{proof}
Let $N=2^f$, $K=2^{f+2}-2$. 
W.l.o.g.\ we assume $n$ is such that $\alpha=n/K$ is integral. 
We give an incompressibility argument by encoding $\alpha$ binary matrices of dimension $N$ in the fault-tolerant diameters of an $n$-vertex directed graph $G$ prone to $2f$ failures. 

Consider any set of $\alpha$ non-null binary matrices $X_1,\ldots,X_\alpha$ of size $N\times N$.
For each $X_i$ construct a subgraph $G_i$ on $K$ vertices as follows.
Take two complete binary trees $T_{i,\ell},T_{i,r}$ of height $f$ rooted at nodes, say, $s_{i},t_{i}$.
Let $\ell_{i,1},\ldots,\ell_{i,2^f}$ be the leaves of tree $T_{i,\ell}$,
and $r_{i,1},\ldots,r_{i,2^f}$ be the leaves of $T_{i,r}$.
We assume edges in tree $T_{i,\ell}$ are directed away from $s_i$.
Conversely, edges in $T_{i,r}$ are pointed towards $t_i$.
Next add edge $(\ell_{i,j_1},r_{i,j_2})$ to $G_i$ if and only if $X_i[j_1,j_2]=1$.
Since $X_i$ is non-null, $t_i$ is reachable from $s_i$ in $G_i$.
The main graph $G$ on vertex set $V = \bigcup_{i=1}^\alpha V(G_i)$ is constructed
by adding a directed edge from $t_{i-1}$ to $s_i$, for each $2 \le i \le \alpha$,
and adding an edge from $t_\alpha$ to each vertex $v \in V$.
See \autoref{figure:diam_fin_str} for an illustration.

We will show next how to extract the matrices $X_i$ by querying the data structure 
with subsets $F\subseteq E$ of cardinality $2f$
on whether or not the diameter graph $\G-F$ is finite.
Consider an index $1 \le i \le \alpha$ and a pair $(j_1,j_2)\in [1,N]\times [1,N]$.
Let $P_\ell$ be the path from the root $s_i$ to the leave $\ell_{i,j_1}$ in the tree $T_{i,\ell}$
and $P_r$ be the path from $r_{i,j_2}$ to $t_i$ in $T_{i,r}$.
Let $F_\ell$ be the set of all edges $(u,v)\in E(T_{i,\ell})$ such that $u\in V(P_{\ell})$ and $v$ is the child of $u$ \emph{not} on $P_{\ell}$.
Define the set $F_r$ analogously with respect to path $P_r$ in $T_{i,r}$.
Set $F=F_\ell\cup F_r$, it contains $2f$ edges of $G$ in total.

Let $P \circ Q$ denote the concatenation of paths $P$ and $Q$.
If $X_i[j_1,j_2]=1$ holds, then in the subgraph $G_i-F$, the unique path from $s_i$ to $t_i$ 
is $P_\ell\circ (\ell_{i,j_1},r_{i,j_2})\circ P_r$.
This in turn implies that $G-F$ is strongly connected (that is, it has finite diameter).
Moreover, if $X_i[j_1,j_2]=0$, then there is no path from $s_i$ to $t_i$ in the subgraph $G_i-F$,
whence $\diam(G-F) = \infty$.

This shows that by just querying whether or not the diameter of graph is finite on occurrence of $2f$ failures, we can extract all the $\alpha$ non-null binary matrices of size $N\times N$.
There are $(2^{N^2}-1)^\alpha$ many such combinations,
with $N = 2^f$ and $\alpha = \Omega(n/2^f)$ this gives the lower bound of $\Omega(2^f n)$ bits.
\end{proof}

\subsection{Eccentricity Oracles for Multiple Failures on DAGs.}
\label{subapp:ecc_oracle}

We study here the problem of single-source eccentricity approximation on directed acyclic graphs. Our focus in on the single-source setting as at most one vertex has a finite eccentricity in a DAG. The objective is to preprocess an $n$ vertex edge-weighted DAG $G=(V,E,s \in V)$ to compute an oracle of polynomial-in-$n$ size that for any given query set $F$ of $f$ edge failures, reports the eccentricity of $s$ in $G-F$ upto a  stretch factor at most $f$, in $O(f)$ time.

Let $T$ be a directed shortest-path tree rooted at $s$, and $\ecc_G(s)$ be $\max_{v\in V}d(s,v)$. For $(x,y)\in E$, let $wt^*(x,y)=d(s,x)+wt(x,y)-d(s,y)$. Observe that for each $e$ in $G$, if $e$ is a tree edge then $wt^*(e)=0$, and if $e$ is a non-tree edge then $wt^*(e)\geq 0$. For each $v\in V$, let $L_{in}(v)$ be the set of first $(f+1)$ non-tree in-edges of $v$ in $G$ having least weight with respect to function $wt^*$, sorted in increasing order of weights.

In the next lemma, we state a property of the $wt^*$ function.

\begin{lemma}~\label{lemma:new-weight}
For any path $P$ in $G$ from $s$ to a node $v\in T$, we have: $wt^*(P) = wt(P) - d(s,v)$.
\end{lemma}

\begin{proof}
Let $P$ be equal to $(s=v_0,\ldots,v_\ell=v)$. So $wt^*(P) =~\textstyle\sum_{i=1}^{\ell} \big( wt(v_{i-1},v_i) + d(s,v_{i-1}) - d(s,v_i) \big)$ which is identical to $\textstyle\sum_{i=1}^{\ell} wt(v_{i-1},v_i) - d(s,v)=~ wt(P) - d(s,v)$.
\end{proof}

Now consider a query set $F$ of $f$ failing edges in $G$. Our query algorithm works as follows: First split $F$ into sets $F_0$ and $F_1$ such that $F_0$ contains the tree edges, and $F_1$ contains non-tree edges. Next, for each $(x,y)\in F_0$, compute an edge  $(\tilde x,y)$ of least weight in $L_{in}(y)\setminus F$, if it exists. If $(\tilde x,y)$ exists then set $\phi(x,y)$ to $wt^*(\tilde x,y)$, otherwise set $\phi(x,y)$ to infinity. The total time to compute all the edges $(\tilde x,y)$ is $O(|F_0|+|F_1|)$ since the lists $L_{in}$ are sorted, and we require to scan a list until we find an edge not lying in $F_1$. As no two lists contains the same failing edge, the total time complexity of scanning the lists is $O(|F|)=O(f)$.

Finally, estimate $\widehat{\ecc}_{G-F}(s)$ as $$\ecc_G(s)+\sum_{(x,y)\in F_{0}}\phi(x,y).~$$

\begin{lemma}
$\ecc_G(s)+\sum_{(x,y)\in F_{0}}\phi(x,y)$ lies in the range $[\ecc_{G-F}(s), ~(f+1)\cdot \ecc_{G-F}(s)]$.
\end{lemma}
\begin{proof}
Let us suppose $F_0=((x_1,y_1),\ldots,(x_k,y_k))$, where $y_i$ appears before $y_{i+1}$ in the topological ordering of $G$, for $i<k$. Let $x_0=s$, and for $i\in[0,k]$, $W_i$ be the set of all vertices in graph $T-F$ reachable from at least one vertex in the set $\{x_0,\ldots, x_i\}$. So, $W_k$ must contain all the vertices lying in $T$.

We first show that for each $i\geq 0$, and $w\in W_i$, there is a path from $s$ to $w$ in $G-F$ whose weight with respect to $wt^*$ is at most $\sum_{j\leq i}\phi(x_j,y_j)$. The claim trivially holds for index $i=0$, as $wt^*$ of tree edges is zero. We will now prove the claim for index $i$, assuming that it holds for index $i-1$. Consider the edge $(x_i,y_i)$. Let $(r_i,y_i)$ be edge in $L_{in}(y_i)\setminus F$ satisfying $wt^*(r_i,y_i)=\phi(x_i,y_i)$. Due to acyclicity of $G$ it follows that $r_i$ lies in the set $W_{i-1}$. By our assumption that claim is true for index $i-1$, there exists a path, say $P$, from $s$ to $r_i\in W_{i-1}$ satisfying $wt^*(P)\leq \phi(x_1,y_1)+\cdots+\phi(x_{i-1},y_{i-1})$. So, $wt^*(P\circ (x_i,y_i))\leq \phi(x_1,y_1)+\cdots+\phi(x_{i},y_{i})$. Now note that all vertices in $W_i\setminus W_{i-1}$ are reachable from $y_i$ by edges whose $wt^*$ is zero. This implies that our claims holds for index $i$.

So, for each $v\in T$, there is a path from $s$ to $v$ in $G-F$, say $P_v$, satisfying $wt^*(P_v)$ at most $\sum_{i\leq k}\phi(x_i,y_i)$. By \autoref{lemma:new-weight}, $wt(P_v)$, for $v\in T$, is at most $\ecc_G(s)+wt^*(P_v)\leq \ecc_G(s)+\sum_{i\leq k}\phi(x_i,y_i)$. Therefore, we have $\ecc_G(s)+\sum_{(x,y)\in F_{0}}\phi(x,y)$ is lower bounded by $\ecc_{G-F}(s)$.

Now we show that for each edge $(x,y)\in F_0$, $\phi(x,y)\leq \ecc_{G-F}(s)$. Recall that $\phi(x,y)=wt^*(\tilde x,y)=d(s,\tilde x)+wt(\tilde x,y)-d(s,y)$, and $(\tilde x,y)$ is the edge in set $L_{in}(y)\setminus F$ (and hence also in $\inedges(y)\setminus F$) having minimum weight according to function $wt^*$. Since $G$ is acyclic, we have that $\phi(x,y)+d(s,y)$ is $s$ to $y$ distance in graph $G-((x,y)\cup F_1)$. Thus, $\phi(x,y)\leq \ecc_{G-((x,y)\cup F_1)} (s)\leq \ecc_{G-F}(s)$, for each $(x,y)\in F_0$.  As $|F_0|\leq f$ and $\ecc_{G}(s) \leq \ecc_{G-F}(s)$, we can conclude that the estimate is at most $(f+1)\cdot \ecc_{G-F}(s)$.
\end{proof}

We thus obtain the following result for DAGs.

\eccDAGupperbound*

\section{Derandomizing the SSRP Algorithm of Chechik \& Magen}
\label{app:omissions_derad_existing}

We present another adaption of our technique to derandomize the $\Otilde(m\sqrt{n}+n^2)$-time
Single-Source Replacement Path algorithm for directed graphs by Chechik and Magen~\cite{ChMa19}.
Previous works~\cite{BCFS21SingleSourceDSO_ESA} have already derandomized an SSRP algorithm
for \emph{undirected} graphs by Chechik and Cohen~\cite{ChechikCohen19SSRP_SODA}.
While the two SSRP algorithms share some similarities,
their derandomization needs different ideas.
For simplicitly, assume we are tasked with the same objective as in \autoref{lemma:large_diam},
that is, hitting all replacement paths of a given minimum length
(the actual requirements differ slightly from that and are discussed below).
The replacement paths in undirected graphs are structurally much simpler.
Afek et al.~\cite{Afek02RestorationbyPathConcatenation_journal} showed that, for a single failure,
on any replacement path $P = P(s,t,e)$ there is a vertex $q$ 
such that the subpaths $P[s..q]$ and $P[q..t]$ are shortest paths in the original graph $G$.
The exact vertex $q$ is hard to compute in general, but Bilò et al.~\cite{BCFS21SingleSourceDSO_ESA} 
devised a method to efficiently hit all (sufficiently long) paths of the form $P[s..q]$
and then use the resulting hitting set to also find vertices on the the second parts $P[q..t]$.
However, this approach breaks down in directed graphs because the Afek et al.~\cite{Afek02RestorationbyPathConcatenation_journal} result fails in this case.
Instead, we use the hierarchies built by the HDPH algorithm to derandomize
the SSRP computation also for directed graphs.

Let $T_s$ be a shortest-paths tree of rooted in $s$,
let $t$ be a {\em balanced separator} of $T_s$ such that 
separates $T_s$ into two edge disjoint
sub-trees $S$ and $T$ such that $n/3 \le |V(S)|, |V(T)| \le 2n/3$, and let 
$P$ be the path from $s$ to $t$ in $T_s$. 
A replacement path $P(s,x,e)$ is called {\em jumping} if it uses some vertex $u$ such that $u \in P$ and $u$ is after the edge failure $e$ in the path $P$. A replacement path that is
not jumping will be called {\em departing}.
If the edge $e$ is among the last $2\sqrt{n}$ edges of $P$, then the replacement paths in $G-e$ are found by computing Dijkstra in $G - e$ in total time of $\Otilde(m\sqrt{n}+n^2)$.
In the rest of this section, assume that $e$ is not among the last $2\sqrt{n}$ edges of $P$.

The only randomization used in the SSRP algorithm of Chechik and Magen is choosing a random hitting set that hits a specific subpath of the replacement path in Case 1.2 or Case 4.3 of their algorithm.
Roughly speaking, Case 1.2 handles departing replacement paths $P(s,x,e)$ such that $x \in V(T)$, and Case 4.3 handles jumping replacement paths $P(s,x,e)$, with $x \in V(S)$, that passes through some vertices of $T$. In both cases, let $v_i$ be the last vertex of $P$ before the edge failure that is also  contained in $P(s,x,e)$ and let $v_j$ be the first vertex of $T$ that is also contained in $P(s,x,e)$. By the choice of the failing edge $e$, the subpath of $P(s,x,e)$ from $v_i$ to $v_j$ has a length of at least $2\sqrt{n}$. More precisely, if the failing edge $e$ is not among the last $z$ edges of $P$, then the subpath of $P(s,x,e)$ from $v_i$ to $v_j$ has length of at least $z$.
Both Cases 1.2 and 4.3 are solved by randomly selecting hitting sets that hit the subpath of $P(s,x,e)$ from $v_i$ to $v_j$. In the following we discuss how Chechik and Magen's (randomized) algorithm deals with both cases and then explain how to derandomize the algorithm without asymptotically affecting the time complexity by more than a polylogarithmic factor.

Chechik and Magen's algorithm considers a logarithmic number of sub-paths $\{P_k\}$, where $P_k$ is the sub-path of $P$ induced by
the vertices $\{v \in V(P) : 2^{k+1} \sqrt{n} \ge d(v,t) \ge 2^k \sqrt{n} \}$. $P_0$ is defined as the sub-path of $P$ induced by the last $2\sqrt{n}$ vertices of $P$. For the sake of simplifying, we assume that $\sqrt{n}$ is a power of 2, i.e., $\sqrt{n}=2^h$, for some integer $h$. This implies that $n=2^{2h}$ is also a power of 2. Note that the set of paths $\{P_k\}$ is an edge disjoint partition of $P$, and that $|P_k| = O(2^k\sqrt{n})$. 
For every index $k \ne 0$, Chechik and Magen's algorithm samples a random set $B_k$ of size $\Otilde(\frac{\sqrt{n}}{2^k})$ using a standard sampling technique, i.e., by adding at random every vertex to $B_k$ with probability $\frac{C \cdot \log n}{2^k\sqrt{n}}$ for large enough constant $C>0$.
As already observed, when $e \in P_k$ for $k \ge 0$, the sub-path of the replacement path $P(s,x,e)$ that we want to hit has a length of at least $2^k\sqrt{n}$. As a consequence, such a sub-path is hit w.h.p. by the set $B_k$. 

We introduce the following notation. 
Let $k \in [0, \log n]$,
denote by $P = P(s,t) = (s = v_0, \ldots, v_\ell = t)$, and let $G_P = G - E(P)$. We say that a pair of vertices $u \in V(P), v \in V(T)$ is called \emph{$k$-relevant} if the following conditions hold.


\begin{definition}\label{def:k-relevant-pair}
A pair of vertices $u \in V(P), v \in V(T)$ is called \emph{$k$-relevant} if $d_{G_P}(u,v) > 2^{k+1}$
and for every vertex $u'$ that appears before $u$ along $P$ it holds that $d_{G_P}(u',v) > d_{G_P}(u,v)$.
\end{definition}

We observe that, for every failing edge $e$ at a distance of at least $2^k$ from $t$,  every replacement path $P(s,x,e)$ that belongs to Case 1.2 or Case 4.3 contains a $k$-relevant path.
Thus, it is sufficient to deterministically compute a set of pivots $B_k$
such that, for every pair of $k$-relevant vertices $u \in V(P), v \in V(T)$, there exists a pivot $b \in B_k$ with 
$d_{G_P}(u,v) = d_{G_P}(u,b) + d_{G_P}(b,v)$. We compute $B_k$ for every $k \in [0,\log n ]$ as described and we set $B_k=B_{h+k}$ for every non-negative integer $k$ such that $2^{k+h}=2^k\sqrt{n} \leq n$.
In \autoref{sec:ssrp-hitting-set-short}, we describe how to compute all the set $B_0,\dots,B_{1/2 \log n}$ so as we can set $B_0=B_{1/2 \log n}$.
In \autoref{sec:ssrp-hitting-set-long}, we describe how to compute such a set of pivots $B_k=B_{k-h}$ given the set of pivots $B_{k-1}=B_{k-h-1}$ for every $k \in (1/2 \log n, \log n]$.

\subsection{Computing the Initial Hitting Sets $B_0$}
\label{sec:ssrp-hitting-set-short}

Consider the following assignment of weights $\omega$ to edges of $G$. We
assign weight $\varepsilon$ for every edge $e$ on the path $P$, and
weight $1$ for all the other edges where $\varepsilon$ is a small number such that $0 < \varepsilon < 1/n$. We
define a graph $G^w = (G,w)$ as the weighted graph $G$ with edge
weights $\omega$.
We define for every $0 \le i \le \ell$ the graph $G_i = G
\setminus \{v_{i+1}, \ldots, v_\ell\}$ and the path $\Pi_i = P \setminus
\{v_{i+1}, \ldots, v_\ell\}$. We define the graph $G^w_i = (G_i, w)$ as
the weighted graph $G_i$ with edge weights $\omega$.

The algorithm computes the graph $G^w$ by simply taking $G$ and setting all edge
weights of $P(s,t)$ to be $\varepsilon$ (for some small $\varepsilon$ such
that $\varepsilon < 1/n$) and all other edge weights to be 1.
The algorithm then removes the vertices of
$P$ from $G^w$ one after the other (starting from the vertex
that is closest to $t$). Loosely speaking, after each vertex is removed, the algorithm computes the distances from $s$ to all the other vertices in the
current graph. In each such iteration, the algorithm adds to
$V^w_{2^k}$ all vertices $v$ such that their distance from $s$ in the current graph is
between $2^k$ and $2^k+1$,
and it adds to 
$P^w_{2^k}$ the last $2^k$ edges of the shortest path from $s$ to every such vertex $v$ in the current graph.
Denote the set of vertices and paths obtained at the end of the algorithm by $V_{2^k} = V^w_{2^k}$ and $P_{2^k} = P^w_{2^k}$, respectively.
Unfortunately, we cannot afford running Dijkstra after the removal
of every vertex of $P$ as there might be $n$ vertices on
$P$. To overcome this issue, the algorithm only maintains the vertices that are at
distance at most $2^k +1$ from $s$.
In addition, we observe that
to compute the SSSP from $s$ in the graph after the removal of a
vertex $v_i$ we only need to spend time on nodes that their shortest path from $s$ uses the removed
vertex. It is not difficult to prove (for example, as done in \cite{AlonChechikCohen19CombinatorialRP}) that for these
nodes their distance from $s$ rounded down to the closest integer must increase by at least 1
as a result of the removal of the vertex. Hence, for every
node we spend time on it in at most $2^k+1$ iterations until its distance from $s$ is bigger than $2^k+1$. A similar runtime analysis as the decremental SSSP algorithm of King \cite{King99FullyDynamicAPSP} and Lemma 27 in  \cite{AlonChechikCohen19CombinatorialRP} shows that the runtime of this algorithm is $\Otilde(m2^k)$, as we scan the neighbours of a vertex $v \in G - V(P)$ at most $2^k$ times.
To obtain the set of pivots $B_k$ we use the greedy algorithm on the set of paths $P_{2^k}$, that is, we set $B_k  = \textsf{GreedyPivotSelection}(P_{2^k})$. 
As $P_{2^k}$ contains $O(n)$ paths of length $2^k$ then according to \autoref{lemma:greedy-correctness} it holds that $|B_k|$ is $\Otilde(n/2^k)$, and the greedy algorithm takes $\Otilde(n 2^k)$ 
time. We have proved the following lemma.

\begin{lemma} \label{lemma:ssrp-hitting-set-short}
Let $k \in [0, 1/2 \log n]$, 
the above algorithm computes in $\Otilde(m 2^k)=\Otilde(m\sqrt{n})$ time a hitting set $B_k$
such that for every pair of $k$-relevant pairs $(u,v)$ with $u \in V(P), v \in V(T)$, there exists a pivot $b \in B_k$ with 
$d_{G_P}(u,v) = d_{G_P}(u,b) + d_{G_P}(b,v)$ and $d_{G_P}(u,b) \leq 2^{k+1}$. 
\end{lemma}

\begin{corollary} \label{cor:ssrp-hitting-set-B_0}
We can compute in $\Otilde(m\sqrt{n})$ time a hitting set $B_0$ such that for every pair of $(1/2 \log n)$-relevant pairs $(u,v)$ with $u \in V(P), v \in V(T)$, there exists a pivot $b \in B_0$ with 
$d_{G_P}(u,v) = d_{G_P}(u,b) + d_{G_P}(b,v)$ and $d_{G_P}(u,b) \leq 2\sqrt{n}$. 
\end{corollary}

\subsection{Computing the set of pivots $B_k$, given the set of pivots $B_{k-1}$, for large $k$} \label{sec:ssrp-hitting-set-long}
In this section we describe how to compute, for every $k \in (1/2 \log n, \log n]$, the set of pivots $B_k$, given the set of pivots that was computed in the previous iteration $B_{k-1}$. We recall that $B_{k-h}=B_k$.
In other words, we describe how to obtain the current hitting set using the previous ones.
Recall that the property that we require here from the set of pivots $B_k$, for every $k \in (1/2 \log n, \log n]$, is that for every pair of $k$-relevant vertices $u \in V(P), v \in V(T)$ there exists a pivot $b \in B_k$ such that $d_{G_P}(u,v) = d_{G_P}(u,b) + d_{G_P}(b,v)$. 

We use the following algorithm to compute the set of pivots $B_k$ from the set of pivots $B_{k-1}$. 
Initialize ${\cal P}_k = \emptyset$.
For every vertex $u \in B_{k-1}$ compute a BFS tree $\T_u$ rooted in $u$ in the graph $G_P$. 
Then, for every $v \in V$ such that $d_{\T_u}(u,v) = 2^{k}$ add the path $P_{\T_u}(u,v)$ to ${\cal P}_k$. 
Finally, use the greedy algorithm to compute $B_k \gets \textsf{GreedyPivotSelection}(\mathcal{P}_k)$.
We prove the correctness and analyze the runtime of the above procedure in the following lemma.

\begin{lemma}\label{lemma:inductive_step_pivots_chechik_magen}
Let $k \in (1/2 \log n, \log n]$, assume that $B_{k-1}$ is a set of $\Otilde(n/2^{k-1})$ vertices such that for every $(k-1)$-relevant pair $(u,v)$ with $u \in V(P)$ and $v \in V(T)$, there exists a pivot $b \in B_{k-1}$ with $d_{G_P}(u,v) = d_{G_P}(u,b) + d_{G_P}(b,v)$ such that $d_{G_P}(u,b) \leq 2^{k}$. 

Then the above procedure computes in $\Otilde(\frac{nm}{2^k} + n^2) = \Otilde(m\sqrt{n} + n^2)$ time a set of $\Otilde(n/2^k)$ pivots $B_k$   
such that for every $k$-relevant pair $(u,v)$ with $u \in V(P)$ and $v \in V(T)$, there exists a pivot $b \in B_k$ with 
$d_{G_P}(u,v) = d_{G_P}(u,b) + d_{G_P}(b,v)$ and such that $d_{G_P}(u,b) \leq 2^{k+1}$. 
\end{lemma}

\begin{proof}
We first prove that $B_k = \Otilde(n/2^k)$ and that the runtime of the above procedure is $\Otilde(\frac{nm}{2^k} + n^2)$. 
Note that $\mathcal{P}_k$ contains $\Otilde(|B_{k-1}| \cdot n) = \Otilde(n^2/2^k)$ paths, each path of length exactly $2^{k-1}$, thus by \autoref{lemma:greedy-correctness} it follows that $B_k = \Otilde(n/2^k)$
and the runtime of the $\textsf{GreedyPivotSelection}(\mathcal{P}_k)$ algorithm is $\Otilde(n^2)$. The above procedure also computes a BFS tree rooted in every vertex $b \in B_{k-1}$, this takes additional $\Otilde(mn/2^k)$ time, so the total runtime of the above procedure is $\Otilde(mn/2^k + n^2)$.

Next, we prove the correctness of the above procedure. We assume that 
for every pair of $(k-1)$-relevant pairs of vertices $u \in V(P), v \in V(T)$, there exists a pivot $b \in B_{k-1}$ with 
$d_{G_P}(u,v) = d_{G_P}(u,b) + d_{G_P}(b,v)$ and such that $d_{G_P}(u,b) \leq 2^{k}$,
and we prove that for every pair of $k$-relevant vertices $u \in V(P), v \in V(T)$ there exists a pivot $b \in B_k$ with 
$d_{G_P}(u,v) = d_{G_P}(u,b) + d_{G_P}(b,v)$ and such that $d_{G_P}(u,b) \leq 2^{k+1}$. 
Let $(u,v) \in V(P)\times V(T)$ be a $k$-relevant pair, and let $v'$ be the $2^{k}-\mathrm{th}$ vertex along $P_{G_P}(u,v)$.

Clearly, $(u,v)$ is also a $(k-1)$-relevant pair. Therefore, by induction, there exists a pivot $b \in B_{k-1}$ with 
$d_{G_P}(u,v) = d_{G_P}(u,b) + d_{G_P}(b,v)$ and such that $d_{G_P}(u,b) \leq 2^{k}$. As $d_{G_P}(b,v) \geq d_{G_P}(u,v)-d_{G_P}(u,b) \geq 2^{k+1}-2^k=2^k$, we have that $B_k$ hits w.h.p. the $(u,v)$ $k$-relevant path as well with a pivot $b' \in B_k$ such that $d_{G_P}(b,b') \leq 2^k$. Therefore, $d_{G_P}(u,v)=d_{G_P}(u,b')+d_{G_P}(b',v)$ and $d_{G_P}(u,b')=d_{G_P}(u,b)+d_{G_P}(b',b) \leq 2^k+2^k \leq 2^{k+1}$.
\end{proof}

\autoref{cor:ssrp-hitting-set-B_0} implies that we can compute $B_0$ in $\Otilde(m\sqrt{n}+n^2)$ time by a deterministic algorithm. As $B_0$ satisfies the hypothesis of \autoref{lemma:inductive_step_pivots_chechik_magen}, by induction on $k$, we have that $B_k$ can be computed out of $B_{k-1}$ in $\Otilde(m\sqrt{n}+n^2)$ time. This completes the derandomization of Chechik and Magen's algorithm. In summary, we proved the following theorem.

\singlesourcereplacementpaths*

\end{document}